\documentclass[12pt]{article}
\usepackage[utf8]{inputenc}
\usepackage{amsfonts, amsthm, amsmath, graphicx, enumerate, verbatim, amssymb, mathrsfs,xparse}
\usepackage[authoryear]{natbib}
\pdfoutput=1
\usepackage{setspace}

\graphicspath{ {images/} }

\newtheorem{thm}{Theorem}
\newtheorem{prop*}{Proposition}
\newtheorem{prop}{Proposition}
\newtheorem{cor}{Corollary}
\newtheorem{thm*}{Theorem}
\newtheorem{lemma}{Lemma}
\theoremstyle{definition}
\newtheorem{defin}{Definition}
\newtheorem{assum}{Assumption}

\newtheorem{remark}{Remark}
\usepackage[top=1in, bottom=1in, left=1in, right=1in]{geometry}
\usepackage{booktabs}

\newenvironment{propbis}[1]
  {\renewcommand{\thethm}{\ref{#1}$'$}%
   \addtocounter{thm}{-1}%
   \begin{thm}}
  {\end{thm}}
\newenvironment{proptris}[1]
  {\renewcommand{\thethm}{\ref{#1}$''$}%
   \addtocounter{thm}{-1}%
   \begin{thm}}
  {\end{thm}}
\newenvironment{propquad}[1]
  {\renewcommand{\thethm}{\ref{#1}$'''$}%
   \addtocounter{thm}{-1}%
   \begin{thm}}
  {\end{thm}}

\newtheoremstyle{named}{}{}{\itshape}{}{\bfseries}{.}{.5em}{\thmnote{#3}#1}
\theoremstyle{named}

\title{Virus Dynamics with Behavioral Responses}
\author{Krishna Dasaratha\thanks{Boston University. Email: krishnadasaratha@gmail.com. I am grateful to the editor and two anonymous referees, Daron Acemoglu, Glenn Ellison, Drew Fudenberg, Edward Glaeser, Benjamin Golub, Chang Liu, Roman Pancs, Matthew Rabin, Sarah Ridout, Tomasz Strzalecki, Eduard Talam\`{a}s, Flavio Toxvaerd, Iv\'{a}n Werning, and Alex Wolitzky for helpful comments and especially to Florian Herold for an insightful discussion. Declarations of interest: none.}}

\date{\today}

\begin{document}

\maketitle
\begin{spacing}{1.5}
\begin{center}\textbf{Abstract}
\end{center}
Motivated by epidemics such as COVID-19, we study the spread of a contagious disease when behavior responds to the disease's prevalence. We extend the SIR epidemiological model to include endogenous meeting rates. Individuals benefit from economic activity, but activity involves interactions with potentially infected individuals. The main focus is a theoretical analysis of contagion dynamics and behavioral responses to changes in risk. We obtain a simple condition for when public-health interventions or variants of a disease will have paradoxical effects on infection rates due to risk compensation. Behavioral responses are most likely to undermine public-health interventions near the peak of severe diseases.

\pagenumbering{gobble}
\newpage
\pagenumbering{arabic}
\setcounter{page}{1}

\section{Introduction}

In a severe epidemic such as COVID-19, the spread of a contagious disease depends substantially on behavioral and policy responses. Standard epidemiological models often predict disease paths with and without social distancing, and outcomes can differ dramatically across these scenarios.\footnote{To give a high-profile example, in March 2020 \cite{ferguson2020report} predicted 2.2 million US deaths from COVID-19 without increased social distancing. Their predicted death toll decreased to 1.1-1.2 million in a mitigation scenario, and there were further reductions under a more aggressive suppression scenario. Beyond the differences in death rates, infection dynamics look very different across the scenarios.} Given the large impact of responses such as social distancing on infection dynamics, understanding how much people will interact is an important open question \citep{funk2015nine}.

To analyze the spread of a disease and the accompanying behavioral responses, we adapt the widely-used SIR (Susceptible/Infected/Recovered) epidemiological model to allow individuals to choose activity levels at each point in time.\footnote{Our model can be modified to include a probability of death rather than recovery; this does not change disease dynamics.} Higher activity levels provide economic and social benefits, but also lead to more interactions with potentially infected individuals, especially when the disease prevalence is high and when others are more active. We use this model to ask how changes in the epidemiological environment, such as public health interventions or new variants, affect the infection rate.

When there is substantial social distancing, the influence of public-health interventions or changes in the virus on behavior is an important consideration. For example, interventions that make interactions safer can be undermined if people respond by interacting more. Related effects have been documented in other settings:   automobile accidents \citep{peltzman1975effects}  and endemic HIV/AIDS at steady state \citep{kremer1996integrating}. But infectious disease outbreaks can change rapidly, and little is known about \textit{when} in an outbreak such effects are likely. Our main results characterize when behavioral responses are largest and most important, and especially when qualitative predictions of the standard SIR model are reversed.

While parts of the analysis are guided by epidemiological evidence on COVID-19, the basic intuitions apply more broadly to other epidemics involving substantial social distancing or related protective behaviors. The key property needed for the intuitions to apply is that the disease prevalence changes rapidly. This need not be the case for endemic diseases long after their discovery; as COVID-19 has demonstrated, however, there can be large and important fluctuations in behavior and incentives earlier in an epidemic.
 

The baseline model considers a population of optimizing agents reacting myopically to current circumstances.\footnote{We begin with myopic agents to obtain sharper results and to avoid assumptions of perfect knowledge of the future path of the disease, and then extend the model to allow forward-looking agents.}  Our main interest is in understanding when in an outbreak behavioral responses are strong enough to reverse direct epidemiological effects. We find that an epidemic can be in a state where three counterintuitive comparative statics hold simultaneously:
\begin{itemize}
  \setlength\itemsep{0em}
\item Increasing the disease prevalence decreases the infection rate.
\item Increasing the transmission rate decreases the infection rate.
\item Decreasing the cost of infection increases the flow costs from new infections.
\end{itemize}
Indeed, any of these holds if and only if the others do. We define a \textit{high infection risk} condition that characterizes the disease states where these comparative statics occur. To illustrate the final bullet point, if the high infection risk condition holds, then a partially effective treatment  may actually increase the health costs of infections because interactions will increase (in the absence of social distancing requirements). We can decompose each comparative static into a term capturing the direct effect, which is the same as in the standard SIR model, and a behavioral response term. When the risk of infection is large enough, the behavioral response can outweigh the direct effect.

The high infection risk condition implies an intuitive description of whether and when in the course of a disease behavioral responses are likely to be large. The counterintuitive effects described above require a severe disease and a high prevalence. With a large infection cost and a high prevalence, there is substantial social distancing and so behavioral responses that weaken social distancing are most consequential. As a result, near the peak of a severe outbreak, public-health interventions can be less effective or even counterproductive unless paired with attempts to maintain social distancing. When prevalence is lower, changes that decrease transmission rates are quite effective. In a simple calibration exercise based on COVID-19, we find a high infection risk region with a higher baseline reproductive number of $R_0=7$ but not with a lower value of $R_0=2.5$, suggesting that the direction of behavioral effects varies within a plausible range of epidemiological parameters.\footnote{These values are obtained from \cite{burki2021omicron} and are estimates of  the reproductive number for the original strain and delta variants of COVID-19, respectively.} 

The model remains analytically tractable and new phenomena emerge with heterogeneous populations, such as people facing different costs of infection due to age or health risks. The framework can also accommodate biased matching between types, which can capture underlying social structure or geography (as in \citealp*{birge2022controlling}). We give a formula expressing equilibrium activity levels in terms of the current population and model parameters.  With heterogeneity, high-risk individuals may temporarily avoid any activity because the infection risk from low-risk people is too high. Upon leaving the economy, high-risk individuals are reluctant to reenter and will wait until the disease prevalence is lower than when they stopped their activity. The techniques used to incorporate preference heterogeneity also allow modifications to the disease mechanics, such as imperfect immunity for recovered individuals.

We also extend our basic results on risk compensation to allow more general incentive structures. First, we allow strategic spillovers in social and economic activity:  activity may become more or less desirable when others engage in more social distancing. Strategic complementarities amplify behavioral responses while strategic substitutes weaken behavioral responses. Second, we show that a quadratic functional form assumption in our baseline model can be relaxed. In general, risk compensation behavior depends on the second derivative of the utility of economic activity.  Third, we allow for forward-looking agents who consider future social distancing and future infection risk when evaluating the value of avoiding infection today.

\subsection{Related Literature}

Much of the early work on behavioral epidemiology focuses on diseases different from COVID-19. Several models of HIV/AIDS epidemics allow individuals to choose numbers of sexual partners  (see \citealp*{kremer1996integrating}, \citealp*{auld2003choices}, \citealp*{greenwood2019equilibrium}, and others). This literature treats HIV/AIDS  as an endemic disease at steady state. By contrast, infection dynamics are important for outbreaks such as COVID-19, and can lead to dramatic shifts in behavior.\footnote{Another difference is that HIV/AIDS modeling uses the SI (Susceptible/Infected) rather than SIR model, as infection is an absorbing state. Theoretical analysis of diseases conferring immunity requires a higher-dimensional state space.} A second strand of the literature focuses individuals' binary choices, such as vaccination, rather than more continuous choices of how much to interact with others (e.g., \citealp*{geoffard1997disease}, \citealp*{galvani2007long}, and several chapters of \citealp*{manfredi2013modeling}). This is especially relevant for less severe or less contagious diseases, such as influenza in ordinary years, for which vaccination choices are more relevant than social distancing.

A number of theoretical papers have studied endogenous social distancing over the course of an epidemic. Early theoretical models include \cite{fenichel2011adaptive},  \cite{chen2012mathematical}, and \cite*{fenichel2013economic}, while more recent work includes \cite*{toxvaerd2019rational}, \cite*{toxvaerd2020equilibrium}, \cite*{mcadams2020nash}, \cite*{mcadams2023equilibrium}, \cite*{carnehl2022time}, and \cite*{engle2021behavioral}. We emphasize three main differences. First, our focus is on how behavioral responses vary with risk levels, and especially on responses to policy changes or shocks that make interactions safer.\footnote{\cite*{toxvaerd2019rational} provides a numerical example where risk compensation can reduce welfare with decentralized behavior and shows that reducing infection risk cannot be harmful if a social planner sets behavior.} We aim for as complete a characterization as possible of the relevant derivatives. In contrast, past work focuses on the basic structure of incentives and externalities. Second, we allow heterogeneity in epidemiological characteristics such as disease risk and in economic preferences, while most theoretical models of social distancing only consider the homogeneous case. Third, we allow strategic complements or substitutes in interaction (as in \citealp*{mcadams2020nash} and \citealp*{mcadams2023equilibrium}), and find an intuitive relationship between these economic feedback effects and epidemiological outcomes.

The main results on risk compensation are closest to several recent papers. \cite*{carnehl2023epidemics} ask how changing the cost of distancing or transmissibility throughout a pandemic affects the peak and total number of infections. To obtain these global comparative statics, their model restricts to distancing by only susceptible agents; the present work focuses on local comparative statics, which allow both temporary and permanent changes in parameters, in a model with richer strategic interactions. \cite*{vellodi2021targeting} also study risk compensation choices in a more abstract model with a one-time choice of how much to interact.\footnote{A similar approach is used by \cite*{acemoglu2016network} to model security in computer networks and \cite*{acemoglu2023testing} to study COVID-19 testing. Both assume agents choose a fixed level of protection and then a virus spreads over time.} Taking local comparative statics in a dynamic model lets us focus on when in an outbreak behavioral responses are most important.

In response to COVID-19, a number of economists have incorporated endogenous behavior into calibrated numerical analyses of SIR models. These models, such as \cite*{farboodi2021internal}, \cite*{eichenbaum2021macroeconomics}, \cite*{krueger2020macroeconomic}, and \cite*{korinek2020covid}, focus largely on macroeconomic consequences of epidemics, and therefore include additional features such as  production sectors or tradeoffs between labor and leisure. While we also include simulations, our focus is on the theoretical properties of the SIR model with endogenous behavior. A major takeaway from our analysis is that basic qualitative properties of disease dynamics and impulse responses can be reversed by changes in model parameters. A simple calibration exercise shows that these reversals can occur within a range of parameters seen as plausible in epidemiological literature on COVID-19.

\section{Model}

\subsection{Setup}

We consider a model in continuous time, indexed by $t \in [0,\infty)$, with a unit mass of individuals. Individuals can be susceptible, infected, or recovered. We denote the shares of susceptible, infected, and recovered individuals at time $t$ by $S(t), I(t)$, and $R(t)$, respectively. We will refer to the shares $(S(t),I(t),R(t))$ as the population disease state.

Each individual chooses a level of activity $q(t) \geq 0$ to maximize their expected flow payoffs at time $t$. The flow payoffs from activity are $$q(t) -a q(t)^2,$$
where $a > 0$. This captures economic and social value from actions that can involve meeting other people, such as going to work or school. The marginal value of these actions is decreasing, and absent infection individuals will choose the finite level of activity $\overline{q} = \frac{1}{2a}$. A susceptible individual pays instantaneous cost $C\geq 0$ if infected.\footnote{We assume infection costs are constant over time. For an analysis of time-varying infection costs, see \cite*{carnehl2022time}.} Quadratic flow payoffs are not essential, and we extend our analysis to more general flow payoffs and forward-looking agents in Section~\ref{sec:gen}.

We now describe the infection process.  When levels of activity are $(q_S(t),q_I(t),q_R(t))$, an individual choosing level of activity rate $q_i(t)$ meets susceptible individuals at Poisson rate $q_i(t)q_S(t)S(t)$, infected individuals at rate $q_i(t)q_I(t)I(t)$, and recovered individuals at rate $q_i(t)q_R(t)R(t)$.\footnote{Since individuals have a unique optimal action $q_i(t)$ for each $t$, we can assume without loss of generality that individuals with the same infection status choose the same level of activity.} That is, individuals meet according to \textit{quadratic matching} \citep*{diamond1979equilibrium}. All meetings are independent.

When a susceptible individual meets an infected individual, the susceptible individual becomes infected with transmission probability $\beta>0$. Infected individuals recover at Poisson rate $\kappa$, and then are immune in all future periods.

When $C=0$, all individuals choose $q(t)=\overline{q}$ and the model reduces to the standard SIR model \citep*{kermack1927contribution}. For $C>0$, the meeting rates between individuals will depend on choices of $q(t).$

\subsection{Equilibrium}

A sequence of (symmetric) action profiles is given by $(q_S(t), q_I(t),q_R(t))_{t=0}^{\infty}$.

Recovered individuals know that they are recovered. Susceptible and infected individuals know that they have not recovered but not whether they are infected, and therefore believe  they are susceptible with probability $\frac{S(t)}{S(t)+I(t)}$ and infected with probability $\frac{I(t)}{S(t)+I(t)}$ at time $t$. Our interpretation is that infected individuals are pre-symptomatic but can spread the disease.\footnote{The model can be extended to allow a fraction of infected individuals who do not develop symptoms, and therefore do not know that when they recover. Introducing asymptomatic infections changes individuals' beliefs about their probability of being susceptible.} (We will discuss a model with equivalent dynamics which explicitly includes symptomatic individuals in Remark~\ref{rem:symptoms}.)

These beliefs imply that susceptible and infected individuals know the sizes of the populations of susceptible and infected individuals but not their current infection status. For example, the sizes of these populations could be calculated from surveillance testing or other publicly available public health data. In our calibrations as well as COVID-19 data prior to the introduction of vaccines, the probability such an individual is susceptible $\frac{S(t)}{S(t)+I(t)}$ is close to one and so behavior would not be very sensitive to some misspecification about the size of the infected population.

A property of the model worth mentioning is that agents pay the infection cost $C$ before becoming aware of the infection. The instantaneous infection cost should be interpreted as a useful modeling technique rather than a literal description of the timing of costs, as delayed costs of infection would not induce social distancing with myopic agents. One could allow infection costs later in the course of the disease in the extension with forward-looking agents (Section~\ref{sec:fwa}).

Susceptible and infected individuals face the same decision problem. So, it will be without loss of generality to consider sequences of actions such that $q_S(t)=q_I(t)$ for all $t$, and any such sequence is characterized by $(q_S(t),q_R(t))_{t=0}^{\infty}$.

An equilibrium is a sequence of action profiles and population disease states with individuals choosing activity levels to maximize flow payoffs at all times, given the population disease state and others' actions, and the disease states following the dynamics described above:
\begin{defin}
An \textbf{equilibrium} given initial conditions $(S(0),I(0),R(0))$ is a sequence of action profiles and shares $(q_S(t), q_R(t), S(t), I(t), R(t))_{t=0}^{\infty}$ such that for all $t \geq 0$
\begin{equation}\label{eq:equilibriumS}
q_S(t)=\text{argmax}_q \left\{q -a q^2 - C q q_S(t) \beta \cdot \frac{I(t)S(t)}{S(t)+I(t)}\right\},
\end{equation}
\begin{equation}\label{eq:equilibriumR}
q_R(t)=\text{argmax}_q\left\{ q -a q	^2 \right\},
\end{equation}
\begin{equation}\label{eq:dynamicsS}\dot{S}(t)= - q_S(t)^2 \beta S(t)I(t),\end{equation}
\begin{equation} \label{eq:dynamicsI}\dot{I}(t)= q_S(t)^2 \beta S(t)I(t) - \kappa I(t),\end{equation}
\begin{equation}\label{eq:dynamicsR} \dot{R}(t)=\kappa I(t).\end{equation}
\end{defin}

We now explain the infection costs in the definition. An individual $i$ choosing level of activity $q_i(t)$ at time $t$ meets infected individuals at rate $q_i(t)q_S(t) I(t)$, and meetings with susceptible and recovered individuals are not relevant for $i$'s payoffs. If individual $i$ is not yet recovered, then $i$ believes they are susceptible with probability $\frac{S(t)}{S(t)+I(t)}$ and already infected otherwise. Since the transmission rate is $\beta$, the susceptible individual would be infected by a meeting with an infected individual at time $t$ with probability $\beta \cdot \frac{S(t)}{S(t)+I(t)}$. The expected rate of infection is therefore $q_i(t) q_S(t) \beta \cdot \frac{I(t)S(t)}{S(t)+I(t)}$, as in equation (\ref{eq:equilibriumS}). If the individual is recovered, there are no infection costs, which gives equation (\ref{eq:equilibriumR}).

\begin{remark}\label{rem:symptoms}
Under the quadratic matching technology, the behavior of recovered individuals does not affect equilibrium behavior by susceptible individuals or the infection rate. Thus, our model is equivalent to a number of extensions with additional heterogeneity among individuals who cannot infect others or be infected.\footnote{These modifications of course have welfare consequences, but behavior and infection dynamics are equivalent.} As one example, infected individuals could be replaced by presymptomatic and symptomatic individuals, where symptomatic individuals choose $q(t)=0$ to avoid spreading the disease.\footnote{Individuals who are not yet symptomatic do not know whether they are susceptible or presymptomatic, and form beliefs as in the baseline model. Presymptomatic individuals become symptomatic at Poisson rate $\kappa$ and symptomatic individuals recover at some Poisson rate. At each time $t$, the infected population in the baseline model matches the presymptomatic population in the modified model and the recovered population in the baseline model matches the combined symptomatic and recovered populations in the modified model.} A second example is including a death rate for infected individuals, so that there are two absorbing states.
\end{remark}

\subsection{Basic Dynamics}\label{sec:dynamics}

We next describe the basic dynamics of the disease. Consider an initial population with $I(0)>0$ and $R(0)=0$. We will be most interested in the dynamics of $(S(t), I(t), R(t))$ starting from a small initial prevalence $I(0)$. We will also often consider the infection rate $\iota(t),$ which is equal to $-\dot{S}(t)$.

Recovered individuals choose $q_R(t) = \overline{q}$ for all $t$, where $\overline{q}=\frac{1}{2a}$. Therefore, an equilibrium is characterized by the level of activity $q_S(t) \leq \overline{q}$ for susceptible and infected individuals, which we will refer to as $q(t)$ from now on.

An individual $i$ who is susceptible or infected chooses $q_i(t)$ to maximize
$$q_i(t) -a q_i(t)^2 - C \beta q_i(t) q(t)  \cdot \frac{I(t)S(t)}{S(t)+I(t)}.$$
Taking the first-order condition for $q_i(t)$ and substituting $q_i(t)=q(t)$, we obtain 
\begin{equation}\label{eq:activity}q(t)=\frac{1}{2a+C\beta \cdot  \frac{I(t)S(t)}{S(t)+I(t)}}\end{equation}
The first-order condition implies there exists a unique symmetric equilibrium given any initial conditions, and the subsequent results characterize the dynamics under this equilibrium.

Equation \ref{eq:activity} implies that the level of activity $q(t)$ is decreasing in the susceptible population $S(t)$ and the prevalence $I(t)$ (drawing the additional susceptible or infected individuals from the recovered population). If the infected population is larger, then activity is more dangerous because individual $i$ is more likely to meet infected peers. If the susceptible population is larger, then individual $i$ is less likely to already be infected.

The baseline reproductive number $R_0=\frac{\beta \overline{q}^2 }{\kappa}$ is the expected number of others who an infected individual would infect in the absence of behavioral responses or immunity. When $R_0$ is less than one, the prevalence $I(t)$ is monotonically decreasing for all $t$ given any initial prevalence and any cost $C \geq 0$. In this case, a small number of cases cannot lead to a larger outbreak. The remainder of the paper analyzes the case in which the transmission probability is high enough for the prevalence to increase (given a small enough initial infected population). Formally, we assume:

\begin{assum}[High enough transmissivity]\label{assumption} $R_0=\frac{\beta\overline{q}^2  }{ \kappa}>1$.
\end{assum}

We describe several basic properties of the model in Appendix~\ref{sec:basic}. In the leading case, disease prevalence is single-peaked under decentralized behavior: the prevalence initially increases, but then peaks and falls toward zero as behavioral responses and herd immunity slow the spread. Part of the population will eventually be immune while others will never be infected, and the relative sizes of these two groups depend on model parameters and potential policy interventions.


\section{Risk Compensation}\label{sec:risk}

We next ask when changes in infection risk have counterintuitive effects due to large changes in behavior. This section shows these counterintuitive effects arise precisely when a high-infection risk condition holds. We begin by defining this condition:

\begin{defin}\label{def:hir}
The \textbf{high infection risk} condition holds at time $t$ if \begin{equation}\label{eq:hir}C \beta \cdot \frac{S(t)I(t)}{I(t)+S(t)}>2 a.\end{equation}
\end{defin}
The left-hand side of the high infection risk condition measures the effect of the behavioral response to a change in infection risk while the right-hand side measures the direct effect of the change. We will see that behavioral responses to policies are most important when there is high infection risk. Expression (\ref{eq:hir}) requires a high cost of infection and large populations of susceptible and infected individuals, so the condition will only hold for severe and highly infectious diseases. Section~\ref{sec:cal} gives parameter values such that the high infection risk condition holds near the peak of the epidemic as well as parameter values for which the condition never holds.

\subsection{Increased Prevalence}

The first consequence of high infection risk is that infection rates are decreasing in disease prevalence, holding the recovered population fixed.

Let $(S^x(t),I^x(t),R^x(t)) = (S(t)-x,I(t)+x,R(t))$ and let $\iota^x(t)$ be the corresponding infection rate. The parameter $x$ can be interpreted as capturing exogenous differences in the population disease state due to shocks such as travel or variation in initial conditions when the disease is discovered. Our first result describes how increasing disease prevalence by increasing $x$ changes the infection rate.

\begin{thm}\label{thm:prevalenceIR}
Suppose there are more susceptible individuals than infected individuals ($I(t)< S(t)$) at time $t$. Then the infection rate is decreasing in the disease prevalence, i.e., $$\left. \frac{\partial \iota^x(t)}{\partial x}\right|_{x=0}<0,$$
if and only if the high infection risk condition holds at time $t$.
\end{thm}

During severe outbreaks, a larger number of individuals infected near the peak can lead to more social distancing. The theorem says that this is large enough to decrease infection rate precisely when the high infection rate condition holds.

The proof shows that $$\left. \frac{\partial \iota^x(t)}{\partial x}\right|_{x=0}= \frac{S(t)-I(t)}{(2a+C \beta\frac{S(t)I(t)}{I(t)+S(t)})^2} \cdot \left(\underbrace{2a}_{\text{direct effect}} -\underbrace{C \beta \cdot \frac{S(t)I(t)}{S(t)+I(t)}}_{\text{behavioral effect}}\right).$$
Increasing the disease prevalence has two effects. First, the positive term corresponds to the direct effect of the change in the infection rate holding behavior fixed, which is positive whenever $I(t) < S(t)$.\footnote{When $I(t)> S(t)$, instead $\left. \frac{\partial \iota^x(t)}{\partial x}\right|_{x=0}>0$ if and only if there is high infection risk at time $t$.} But second, the negative term corresponds to the effects of the change in behavior in response to the increased infection. A higher prevalence leads to a lower level of activity $q(t)$ (when $I(t) < S(t)$), which lowers the infection rate. Condition~(\ref{eq:hir}) determines whether the behavioral effect is larger than the direct effect. By contrast, the standard SIR model only captures the direct effect.

A similar effect occurs in the steady state model of \cite{quercioli2006contagious}, where individuals meet at an exogenous rate but can take protective actions. In a model where disease prevalence changes over time, we find this comparative static coincides with other risk compensation effects: we next show the  high infection risk condition also determines whether behavioral responses outweigh the direct effects of several policy changes.

\subsection{Severity and Transmission Rate}

We next consider the effects of changes in the cost $C$ from infection and the transmission rate $\beta$, which can correspond to public-health measures that change the spread of the disease or the arrival of a variant with higher transmissivity or lower severity. For example, a more contagious variant would correspond to an increase in $\beta$ while a partially effective treatment would decrease $C$.

Our high infection risk condition again determines the sign of the impact of these policies:

\begin{thm}\label{thm:policy}
The following are equivalent:
\begin{enumerate}[(1)]
\item The high infection risk condition holds at time $t$,
\item Marginally decreasing the transmission rate $\beta$ at time $t$  increases the infection rate $\iota(t)$, and
\item Marginally decreasing the infection cost $C$ at time $t$ increases the flow costs $C\iota(t)$ from infections.
\end{enumerate}
\end{thm}

The theorem says that these policies to decrease $\beta$ and $C$ increase short-term flow costs from infection  in the same high-infection risk region, which is characterized by high prevalence and a high cost of infection. Formally, (2) and (3) state the inequalities $$\frac{\partial \iota(t)}{\partial \beta} <0\text{ and }  \frac{\partial (C\iota(t))}{\partial C}<0 $$
at the initial values of $\beta$ and $C$. As in the proof of Theorem \ref{thm:prevalenceIR}, we can determine the signs of these derivatives by decomposing the effects of changes in $\beta$ or $C$ into a direct effect and a behavioral effect. The high infection risk condition corresponds to the behavioral effect being larger.

The model suggests that considering behavioral responses to public health policies is most important for severe diseases and close to the peak of outbreaks. Near the peak, public health interventions will tend to be less effective or even counterproductive unless paired with campaigns or restrictions to encourage social distancing. When prevalence is low, public health measures such as requiring masks are effective alone.

Related steady-state and one-shot models find that partially effective vaccines can increase infection rates. This was first observed by \citep*{kremer1996integrating} for HIV/AIDS. \cite*{toxvaerd2019rational} and \cite*{talamas2020free} show partially effective vaccines can also decrease welfare due to the negative externalities from infection. We provide a general framework that determines when in an outbreak several counterintuitive comparative statics, including the effect from \cite*{kremer1996integrating}, will occur for contagious diseases such as COVID-19. We find that partially effective vaccines will increase infection rates precisely when a lower disease prevalence increases infection rates (which cannot occur at steady state) and when partially effective treatments increase flow costs from infections.

\subsection{Calibration}\label{sec:cal}

We now perform a simple calibration exercise to give intuition for when there will be a high infection risk. The purpose of the exercise is to explore how the qualitative properties of behavioral responses depend on model parameters. We find there are qualitative differences within reasonable parameter ranges: high infection risk is unlikely under estimates of $R_0$ for the original strain of COVID-19 but does hold near peak prevalence in simulations with an estimated $R_0$ for the delta variant.

The model is specified (up to normalizations) by the basic reproductive number $R_0=\frac{\beta \overline{q}^2 }{\kappa}$ before immunity or behavioral responses, the average disease length $\frac{1}{\kappa}$, and a preference parameter $\frac{C}{4a}$.  We briefly discuss each of these parameters and the values used in our simulations:

\begin{enumerate}[(1)]
\item The parameter $R_0$ is purely epidemiological. We describe simulations under values of $R_0^{low} = 2.5$ and $R_0^{high}=7.0$, which are based on parameter values for the original strain and delta variants of COVID-19 from \cite{burki2021omicron}.\footnote{If these values are based in part or fully on models that do not account for social distancing they may be biased downward. Higher values of $R_0$ would strengthen behavioral responses in our simulations.} Simulations suggest behavioral responses are more sensitive to $R_0$ than the other two parameters. The simulation results below can be interpreted as giving a sense of how large the reproductive number $R_0$ needs to be for the high infection risk condition to hold for infections severe enough to induce substantial social distancing.

\item The parameter $\kappa$ is also purely epidemiological. We set the average disease length $\frac{1}{\kappa}$ to be $14$ days, which is approximately twice the estimated serial interval between infections \citep{sanche2020high}. Some work has estimated substantially shorter infectious periods, e.g., 3.5 days in \cite{li2020substantial}. A shorter average disease length would give stronger behavioral responses in our calibration.

\item  The preference parameter $\frac{C}{4a}$ is the sole economic parameter in the calibration, and measures the cost of infection divided by the consumption utility from one period of normal activity. To see this, recall that without social distancing $\overline{q} = \frac{1}{2a}$ and the corresponding flow payoffs are $\overline{q}-a\overline{q}^2=\frac{1}{4a}.$

We assume the  cost of infection is equal to the flow payoffs consumption utility from six months of normal activity. This is comparable to existing parameter values in macroeconomic models of COVID-19, many of which assume an infection cost ranging from three to twelve months of the value of US per capita consumption. These parameters are often derived from estimates of the value of a statistical life and the infection-fatality ratio for COVID-19 (see \citealp*{farboodi2021internal} for a discussion). The results we now discuss also remain qualitatively unchanged for other values between three and twelve months.
\end{enumerate}

\begin{figure}
\center{\includegraphics[scale=.7]{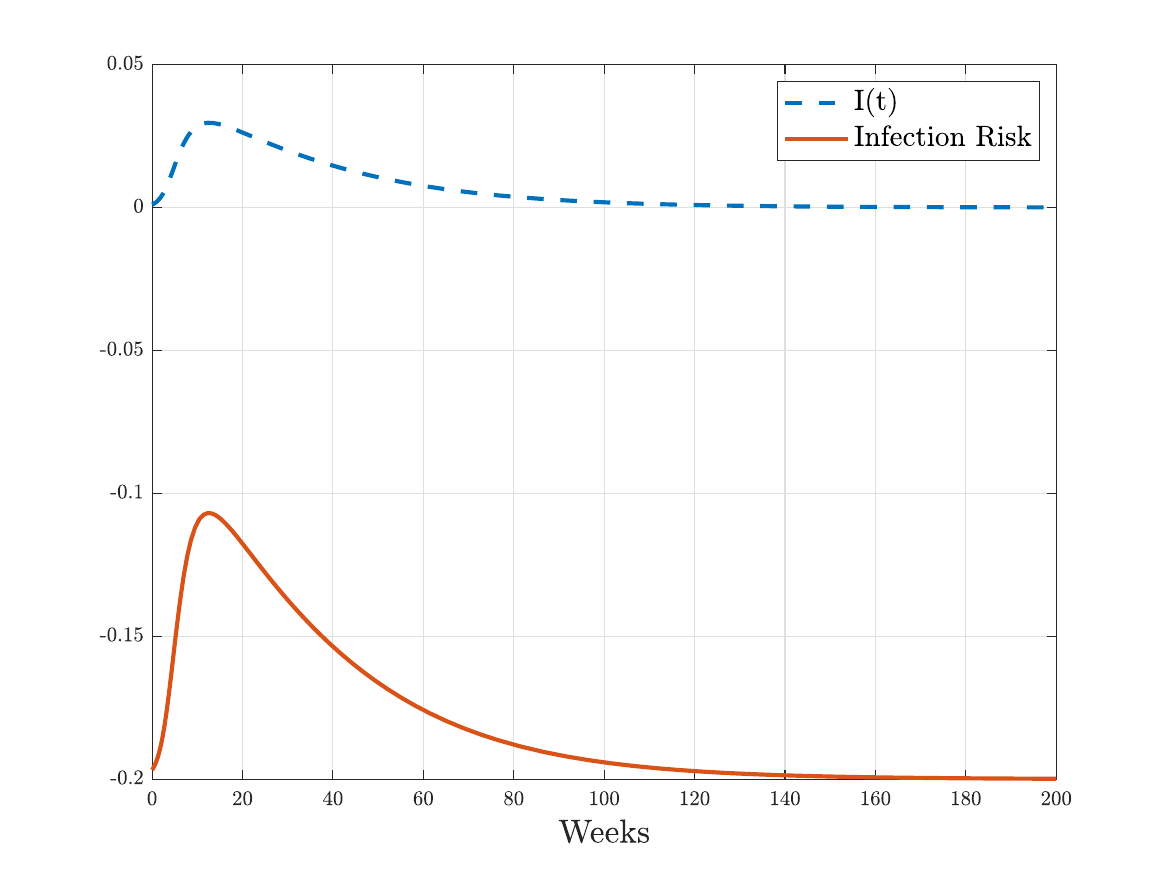}

\includegraphics[scale=.7]{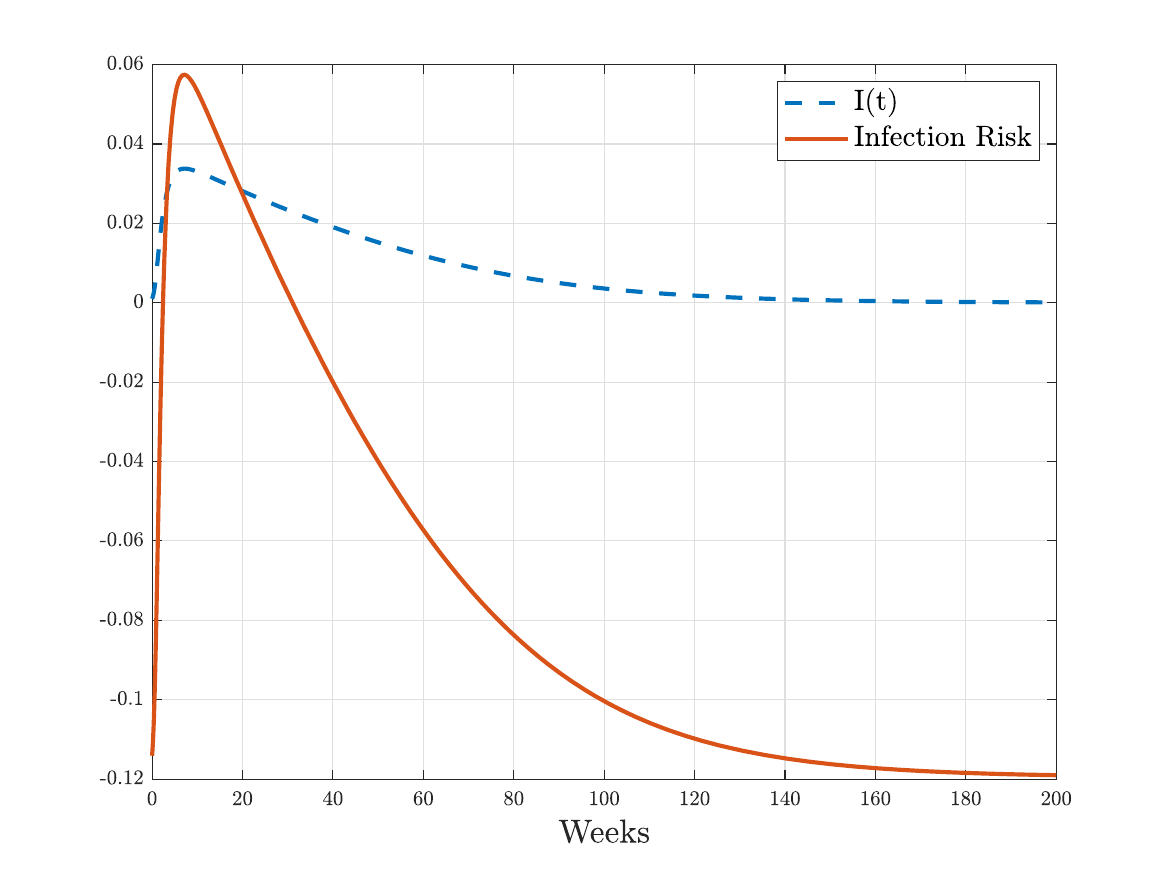}}\caption{Time path of $I(t)$ in dashed blue and $2a-C \beta \cdot \frac{S(t)I(t)}{I(t)+S(t)}$ in solid red for $R_0^{low}=2.5$ (top) and $R_0^{high}=7.0$ (bottom). The high infection risk condition holds when the value of the solid red curve is positive. The average disease length is $14$ days and the infection cost is equal to the value of six months of normal activity.}\label{fig:riskcomp}
\end{figure}

Figure~\ref{fig:riskcomp} shows the infection paths in dashed blue and the infection risk $2a-C \beta \cdot \frac{S(t)I(t)}{I(t)+S(t)}$ in solid red for each of the two values of $R_0$. The high infection risk condition holds when the solid red curve is positive. When the reproductive number is $R_0^{low}=2.5$, there is not a high infection risk region; though behavioral responses still matter for comparative statics, they will be less important than direct effects. When the reproductive number is $R_0^{high}=7.0$, the high infection risk condition holds near and after the epidemic's peak. This suggests that behavioral forces can differ substantially within a plausible range of epidemiological parameters.

This section described a calibration exercise using parameter values from the epidemiology and economics literatures, but an alternate approach would be to determine these values via a structural estimation. One could ask which values of the three model parameters most closely fit data on COVID-19 infections and/or mobility data on social distancing. Such an exercise may be particularly relevant for determining the infection cost $C$ that can best explain empirical behavior.

\subsection{Discussion of Policy Implications}

We now briefly discuss policy implications of our comparative statics results. Consider a public health intervention that decreases the transmission rate $\beta$ or the infection cost $C$. When behavioral responses are small, these interventions will be effective on their own. When behavioral responses are sufficiently large, however, such an intervention will be less effective or even counterproductive in decreasing short-term infections or infection costs. The intervention may still increase welfare in our model, but if decreasing infections or infection costs is the primary policy goal then they will be ineffective, at least alone. One approach is to pair such public health interventions with policies, such as lockdowns, that maintain social distancing; the recent literature on optimal lockdowns (e.g., \citealp*{acemoglu2020multirisk}) sheds some light on the tradeoffs involved.

The high infection risk condition tells us when the short-term behavioral effect of certain policy changes is larger than the direct effect. More generally, the difference $$2a-C\beta \cdot \frac{S(t)I(t)}{S(t)+I(t)}$$
may be a useful measure of the size of these behavioral effects. When this difference is negative but small, behavioral responses will substantially dampen the effects of policy changes. In Section~\ref{sec:gen}, we will see how to extend this measure beyond the baseline model to allow utility functions that are more general in several respects.

\section{Heterogeneity}

We now incorporate heterogeneity in economic preferences and epidemiological characteristics. We begin by characterizing equilibrium actions given an arbitrary type space. Restricting to two types, we explicitly describe the behavior of high-risk and low-risk populations over time.
\subsection{Equilibrium}\label{sec:het}

A feature of our model with quadratic utility is that behavior continues to take a simple analytic form with heterogeneous populations. We first give a procedure for determining equilibrium actions at a given time $t$.

Suppose there are $m$ types with each type $k$ constituting a share $\alpha_k$ of the population. Each type has infection cost $C_k$ and receives flow payoffs $q-a_k q^2.$ Variation in $C$ can capture differences in health risk due to age and comorbidities. Variation in preferences can also distinguish essential and non-essential workers. We will allow for heterogeneous social structure via biased meeting probabilities and spillovers between types below.

For each $k$ we let $S_k(t)$, $I_k(t)$, and $R_k(t)$ be the population share of susceptible, infected, and recovered individuals of type $k$ at time $t$, so that $S_k(t)+I_k(t)+R_k(t)=\alpha_k$. As in the homogeneous case, recovered individuals choose $q_k(t)=\overline{q}$ for all $t$. An equilibrium strategy is now given by levels of activity $(q_k(t))_{k=1}^{m}$ for susceptible or infected individuals of each type such that $$q_k(t)=\text{argmax}_q  \left\{ q-a_kq^2   - C_kq \cdot \left(\sum_{k'=1}^m q_{k'}(t) I_{k'}(t)\right) \beta  \cdot \frac{S_k(t)}{S_k(t)+I_k(t)} \right\}$$
for all $k$ and all $t$.

The maximization problem differs from the homogeneous case in that (1) the infection cost and the probability of being susceptible now differ across types and (2) the probability of meeting infected individuals of another type now depends on the infection prevalence and activity level for that type, as captured by the summation across types $k'$.

In the homogeneous case, the equilibrium activity level $q(t)$ is always positive. If all other individuals chose $q(t)=0$, then activity would be safe and so the best response would be $q(t)=\overline{q}$. A key difference under heterogeneity is that high-risk types may choose $q(t)=0$ because low-risk types are active enough to make any positive level of activity unsafe. That is, if there are enough low-risk infected people in public spaces, high-risk people may stay home entirely.

Consider an equilibrium at time $t$ at which some set of types $A \subset \{1,\hdots,m\}$ choose positive activity levels while other types $k \notin A$ choose $q_k(t)=0$. Then the activity levels of individuals in $A$ are determined by the corresponding first-order conditions. Letting $M_A=\left(C_k\beta \cdot \frac{S_k(t)I_{k'}(t)}{I_k(t)+S_k(t)}  \right)_{ k,k' \in A}$, equilibrium activity is 
\begin{equation}\label{eq:het}(q_k(t))_{k \in A} = \left(M_{A} + 2 \cdot \textnormal{diag}(a_k) \right)^{-1} \bf{1} .\end{equation}
Here $\text{diag}(a_k)$ is the diagonal matrix with entries $a_k$ and $\bf{1}$ is the column vector of ones. 

To determine the set of types $A$ choosing positive actions at equilibrium, we begin with all types in $A_0=\{1,\hdots,m\}$. Given $A_j$, we obtain $A_{j+1}$ by removing from $A_j$ all types $k$ such that the $k^{th}$ entry of $\left(M_{A_j} +  2 \cdot \text{diag}(a_k)\right)^{-1} \bf{1}$ is negative. The following proposition shows that this process stabilizes within at most $m$ steps and that equilibrium activity is given by equation (\ref{eq:het}) with $A=A_m$.

\begin{prop}\label{prop:het}
Equilibrium activity is given by
\begin{equation*}\label{eq:hetprop}(q_k(t))_{k \in A_m} = \left(M_{A_m} + 2 \cdot \mbox{diag}(a_k) \right)^{-1} \bf{1}\end{equation*}
and $q_k(t)=0$ for $k \notin A_m$.
\end{prop}

The proposition gives a procedure for determining which types will choose positive activity, and expresses these types' behavior in terms of model parameters and state variables. Since people play a quadratic game at each time $t$, we obtain a simpler expression for behavior than other functional forms for utility would allow. The procedure repeatedly removes types for whom equation~(\ref{eq:het}) does not give a feasible activity level. The strategic substitutability of actions at time $t$ ensures this process does characterize the equilibrium.

The characterization in equation~(\ref{eq:het}) can be extended to allow biased matching, though determining the set $A$ of types choosing positive activity may be more complicated. For example, we could let individuals of type $k$ meet individuals of type $k'$ in each disease state at rates $\gamma_{kk'}q_k(t)q_{k'}(t)S_{k'}(t)$, $\gamma_{kk'}q_k(t)q_{k'}(t)I_{k'}(t)$, and $\gamma_{kk'}q_k(t)q_{k'}(t)R_{k'}(t)$ for some symmetric matrix $(\gamma_{kk'})_{1\leq k,k'\leq m}$. The weights $\gamma_{kk'}$ allow a biased matching process, with meeting probabilities depending on factors such as geography or social structure as well as activity levels. Then equation~(\ref{eq:het}) holds with $M_A=\left(C_k\beta\gamma_{kk'} \cdot \frac{S_k(t)I_{k'}(t)}{I_k(t)+S_k(t)}  \right)_{ k,k' \in A}$.

The approach we have used to allow heterogeneity in preferences can also be used to modify the set of possible infection statuses. We briefly discuss one example, partial immunity, which is particularly relevant for COVID-19. Suppose that recovered individuals are still vulnerable to interactions with infected individuals, but are infected with some lower probability $0<\beta^R < \beta$. The same approach used in Proposition \ref{prop:het} will characterize the behavior of individuals who have never recovered as well as those who have recovered at least once. Depending on the value of $\beta_R$, which measures the extent of the partial immunity, the infection can become endemic ($\liminf_t I(t) >0$) or eventually die out ($\lim_t I(t) = 0$).

\subsection{Two Types}

One important case is heterogeneity in health risks. We next discuss the dynamics when there are two types with costs $C_1 < C_2$, who we will call young and old individuals respectively. For simplicity we assume unbiased matching and assume that both types have the same flow payoffs $q(t)-aq(t)^2$.

With a small initial prevalence, both types will interact and become infected at positive rates early or late enough in the outbreak. But if the gap between the types' infection costs is large enough, then only young people will choose positive activity levels near the peak. The next result describes when this occurs:

\begin{prop}\label{prop:twotype}
Old individuals choose positive activity level $q_2(t)>0$ if and only if 
$$C_2 \beta  \cdot \frac{I_1(t)S_2(t)}{S_2(t)+I_2(t)} \geq  2a+ C_1 \beta \cdot \frac{I_1(t)S_1(t)}{I_1(t) + S_1(t)}.$$
If old individuals stop choosing $q_2(t)>0$ at time $t_1$ and resume choosing $q_2(t)>0$ again at time $t_2$, then $I_1(t_1) > I_1(t_2).$
\end{prop}

\begin{figure}
\center{\includegraphics[scale=.65]{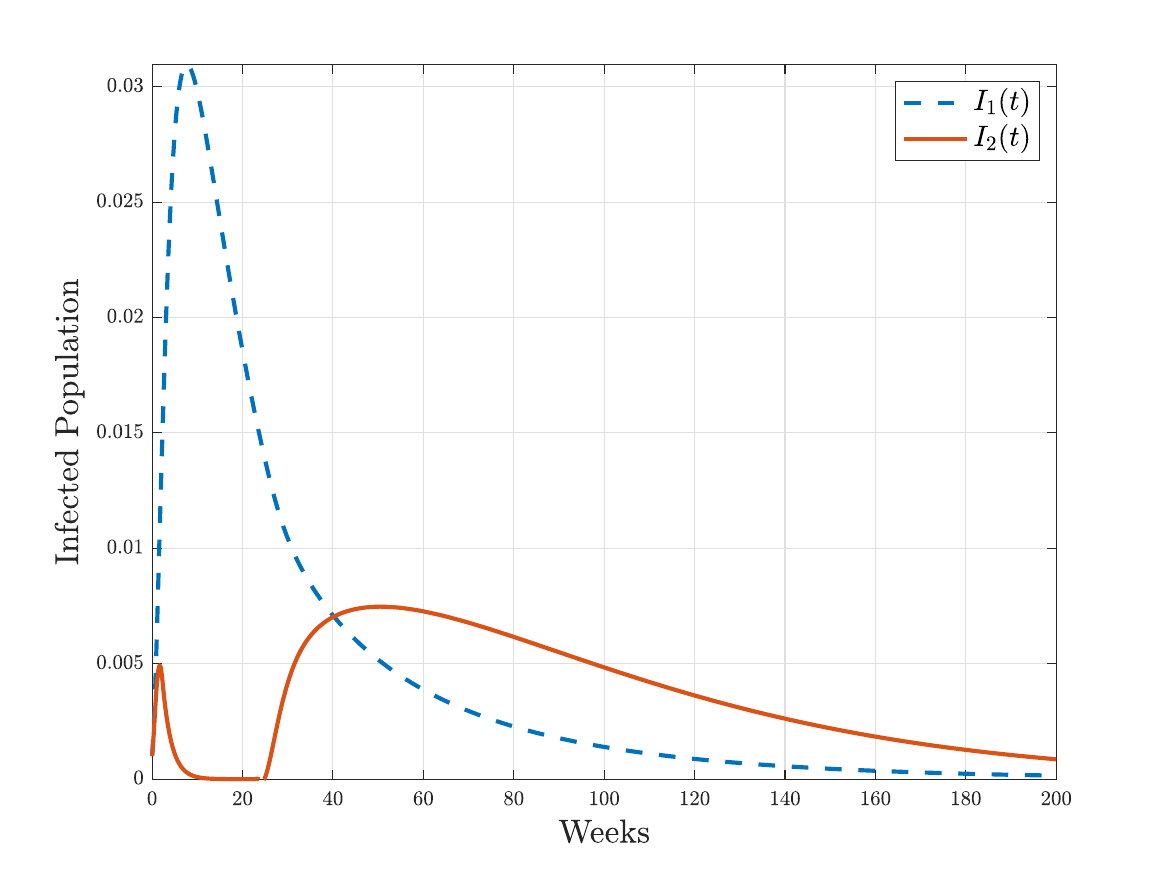}

\includegraphics[scale=.65]{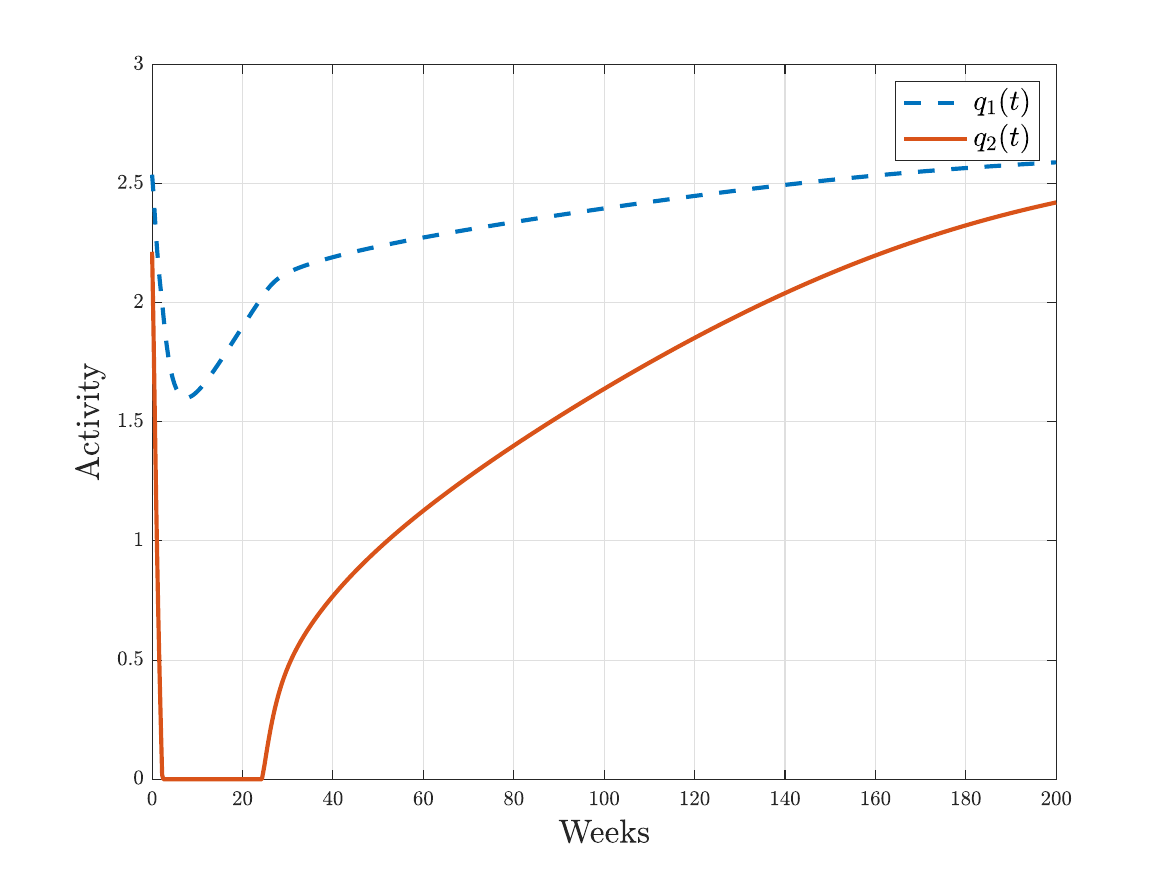}}

\caption{High-risk individuals may choose $q_k(t)=0$ when prevalence is high. Graphs show the time paths of $I_k(t)$ (top) and $q_k(t)$ (bottom) for low-risk indiviudals (dashed blue) and high-risk individuals (solid red). Each group has size $\alpha_k=\frac12$, and infection costs $C_1$ and $C_2$ are equal to the value of three months of normal activity and one year of normal activity, respectively. The reproductive number is $R_0=5.7$ and the average disease length is $14$ days.}\label{fig:twogroup}
\end{figure}

The first part of the proposition describes when old individuals choose zero activity in terms of the population disease state and model parameters. The second part states that if old individuals stop interacting with others at some time, they will not stop again until the disease prevalence among young people is strictly lower than when they stopped.

So if one group leaves the economy, it is more difficult to convince them to resume activity. The main force is that if an individual has not engaged in any activity for a long period (and has not contracted the disease and recovered), that individual is very likely to be susceptible. Thus, infection becomes a bigger concern over time, and returning to the economy requires safer conditions than were acceptable before leaving. The proof shows this intuition continues to hold after allowing for  changes in young people's behavior over time.

Figure~\ref{fig:twogroup} shows an example of the infection rates (top) and levels of activity (bottom) in each group. In the figure, the two groups are equally sized and $I_1(0) = I_2(0)=0.001$ while the remaining initial population is susceptible. The infection costs are equal to the value of three months of normal activity for the low-risk group and one year of normal activity for the high group. We consider the higher value of $R_0=7.0$ from  Section~\ref{sec:cal}.

We have $I_1(t) \approx 0.20$ when old individuals cease activity and $I_2(t) \approx 0.17$ when old individuals resume activity. The threshold infection rate for ceasing activity is higher than the threshold for resuming activity. We also observe that when old individuals resume activity, they do so rapidly. Indeed, when the gap between the costs $C_1$ and $C_2$ is larger than in this example figure, the rapid resumption of activity can lead to a non-monotonicity in $q_2(t)$ as young individuals decrease activity when old individuals return to the economy.

Several recent numerical models have also introduced heterogeneity, with a focus on lockdowns targeted at older individuals (\citealp*{brotherhood2020economic} and \citealp*{acemoglu2020multirisk}). We instead describe the decentralized behavior of each group in a heterogeneous population analytically. One finding under our behavioral specification is that targeted lockdowns may not always be binding, even if these policies are quite strict.

\section{Payoff Structure}\label{sec:gen}

We show that our risk compensation results from Section~\ref{sec:risk}  generalize beyond myopic agents with utility function $q-aq^2$ for economic activity.  We first show that behavioral responses are strengthened by strategic complementarities and weakened by strategic substitutability. Second, we allow the flow payoffs from economic activity to be a general concave function. Third, we consider forward-looking agents with an arbitrary discount rate.

\subsection{Strategic Spillovers}\label{sec:payoffcomp}

We first extend the quadratic functional form to allow activity levels to be strategic complements or substitutes. Economic activity may become more attractive when others social distance, as opportunities for social interaction and exchange decrease, or less attractive, as competition may decrease.

In this subsection, we first allow spillovers in a homogeneous population and then allow spillovers that depend on types. Consider flow payoffs from activity 
\begin{equation}\label{eq:quadcomp}q_i(t)-aq_i(t)^2+b \int_j q_i(t)q_j(t),\end{equation}
where $a>0$ and $b<2a$. The integral is taken over all individuals $j$, and allows incentives to depend on the average action  in the population. Actions are strategic complements for $b>0$ and substitutes for $b<0$.

The definition of equilibrium becomes:
\begin{defin}

An \textbf{equilibrium} given initial conditions $(S(0),I(0),R(0))$ is a sequence of action profiles and shares $(q_S(t), q_R(t), S(t), I(t), R(t))_{t=0}^{\infty}$ such that for all $t \geq 0$
\begin{equation}\label{eq:foccompS}
q_S(t)=\text{argmax}_q \left\{q - a q^2 + bq \cdot \left( \int_j q_j(t) \right)- C q q_S(t) \beta \cdot \frac{I(t)S(t)}{S(t)+I(t)}\right\},
\end{equation}
\begin{equation}\label{eq:foccompR}
q_R(t)=\text{argmax}_q \left\{q - a q^2 + bq\cdot \int_j q_j(t) \right\},
\end{equation}
where the average action $\int_j q_j(t) $ is equal to $(S(t)+I(t))q_S(t) + R(t)q_R(t),$ and equations (\ref{eq:dynamicsS}), (\ref{eq:dynamicsI}), and (\ref{eq:dynamicsR}) hold.
\end{defin}

The equilibrium action $q_R(t)$ of recovered individuals is no longer constant, because the optimal value of $q_R(t)$ depends on $q_S(t)$ and the size of the recovered population. Equations~(\ref{eq:foccompS}) and~(\ref{eq:foccompR}) give linear first-order conditions for $q_S(t)$ and $q_R(t)$, and we can solve for the equilibrium action\footnote{Because we assume $a>0$ and $b<2a$, there exists a unique equilibrium and the activity levels are positive at equilibrium.}
\begin{equation}\label{eq:payoffcompqs}q_S(t) = \frac{1+\frac{b R(t)}{2a-b}}{2a-b(S(t)+I(t))\left(1+\frac{b R(t)}{2a-b}\right) + C\beta \cdot \frac{S(t)I(t)}{S(t)+I(t)}}.\end{equation}

We next extend our risk compensation results from Section~\ref{sec:risk}. To do so, we first define a high infection risk condition that allows for strategic spillovers:

\begin{defin}
The \textbf{high infection risk} condition holds at time $t$ if \begin{equation}\label{eq:hirgen}C \beta \cdot \frac{S(t)I(t)}{I(t)+S(t)}> 2a- b (S(t)+I(t))\left(1+\frac{b R(t)}{2a-b}\right).\end{equation}
\end{defin}
The definition generalizes Definition~\ref{def:hir}. The new term $$- b (S(t)+I(t))\left(1+\frac{b R(t)}{2a-b}\right)$$ on the right-hand side is positive under strategic substitutes and negative under strategic complements. So high infection risk is easier to achieve under strategic complements and harder to achieve under strategic substitutes. 

As in Theorem~\ref{thm:prevalenceIR}, we take $(S^x(t),I^x(t),R^x(t)) = (S(t)-x,I(t)+x,R(t))$ and let $\iota^x(t)$ be the corresponding infection rate. The following two theorems hold, just as in the baseline model.

\begin{propbis}{thm:prevalenceIR}
Suppose $I(t)< S(t)$. Then $\left. \frac{\partial \iota^x(t)}{\partial x}\right|_{x=0}<0$
if and only if the high infection risk condition holds at time $t$.
\end{propbis}

\begin{propbis}{thm:policy}
The following are equivalent:
\begin{enumerate}[(1)]
\item The high infection risk condition holds at time $t$,
\item Marginally decreasing $\beta$ at time $t$  increases $\iota(t)$, and
\item Marginally decreasing $C$ at time $t$ increases $C\iota(t)$.
\end{enumerate}
\end{propbis}

As in Section~\ref{sec:risk}, strong behavioral responses occur in the high infection risk region. The cross partial term $-b (I(t)+S(t))\left(1+\frac{b R(t)}{2a-b}\right)$ from expression~(\ref{eq:hirgen}) implies that the behavioral response to changes in environment or policy is larger when actions are strategic complements and smaller when they are strategic substitutes. Complementarity in actions amplifies behavioral responses, as economic activity becomes less desirable when others engage in more social distancing. Substitutability weakens behavioral responses, as economic activity becomes more desirable when others engage in more social distancing.

Strategic spillovers matter more when $b (I(t)+S(t))\left(1+\frac{b R(t)}{2a-b}\right)$ is larger. This tends to be when $S(t)+I(t)$ is larger, because susceptible and infected individuals respond directly to changes in risk levels while recovered individuals only respond to others' changes. This need not be the case if there are strong enough complementarities for $2a-b$ is small, however.

\subsection{General Flow Payoffs}

We next relax the quadratic utility function from the baseline model. Consider flow payoffs from activity $f(q_i)$, where $f$ is twice differentiable, strictly concave, increasing at zero, and decreasing for $q_i$ sufficiently large.\footnote{We assume that flow payoffs do not depend on others' economic activity, but could also allow strategic spillovers as in Section~\ref{sec:payoffcomp}. The high infection risk condition would then have an additional term depending on the cross partial derivative of $f$ in an individual's own action and the average action. We describe the two generalizations separately for simplicity.}

The definition of equilibrium becomes:
\begin{defin}
An \textbf{equilibrium} given initial conditions $(S(0),I(0),R(0))$ is a sequence of action profiles and shares $(q_S(t), q_R(t), S(t), I(t), R(t))_{t=0}^{\infty}$ such that for all $t \geq 0$
\begin{equation}\label{eq:equilibriumS}
q_S(t)=\text{argmax}_q \left\{ f(q) - C q q_S(t) \beta \cdot \frac{I(t)S(t)}{S(t)+I(t)} \right\}, q_R(t)=\text{argmax}_q f(q),
\end{equation} and equations (\ref{eq:dynamicsS}), (\ref{eq:dynamicsI}), and (\ref{eq:dynamicsR}) hold.
\end{defin}

Because $f'(0)>0$, both the equilibrium actions $q_S(t)$ and $q_R(t)$ are positive for all $t$. The equilibrium action $q_R(t)$ is a constant depending on $f$. The equilibrium action $q(t)=q_S(t)$ of susceptible and infected individuals may not have an analytic solution, but is characterized by the first-order condition
\begin{equation}\label{eq:focgen}f'(q(t))- C q(t) \beta \cdot \frac{I(t)S(t)}{S(t)+I(t)} = 0.\end{equation} The high infection risk condition can again be generalized:


\begin{defin}
The \textbf{high infection risk} condition holds at time $t$ if \begin{equation}\label{eq:hirgen}C \beta\cdot \frac{I(t)S(t)}{S(t)+I(t)} > -f''(q(t)).\end{equation}
\end{defin}
Because $f$ is strictly concave, the right-hand side is positive. The high infection risk condition is more likely to hold when the elasticity of the activity level $-f''(q(t))$ is smaller, i.e., payoffs are more concave in activity levels. In this case individuals will react more to changes in risk levels by increasing or decreasing social distancing.

Our risk compensation results continue to hold:

\begin{proptris}{thm:prevalenceIR}
Suppose $I(t)< S(t)$. Then $\left. \frac{\partial \iota^x(t)}{\partial x}\right|_{x=0}<0$
if and only if the high infection risk condition holds at time $t$.
\end{proptris}

\begin{proptris}{thm:policy}
The following are equivalent:
\begin{enumerate}[(1)]
\item The high infection risk condition holds at time $t$,
\item Marginally decreasing $\beta$ at time $t$  increases $\iota(t)$, and
\item Marginally decreasing $C$ at time $t$ increases $C\iota(t)$.
\end{enumerate}
\end{proptris}

A simple example which nests the quadratic functional form is  $$f(q_i(t)) = q_i(t) -a q_i(t)^{\alpha},$$ where $\alpha>1$. Then the elasticity of activity levels with respect to infection risk is equal to $a\alpha(\alpha-1)q^{\alpha-2}.$ This elasticity is decreasing in $C \beta \cdot \frac{S(t)I(t)}{I(t)+S(t)}$ when $\alpha <2$ and increasing in $C \beta \cdot \frac{S(t)I(t)}{I(t)+S(t)}$ when $\alpha>2$. When $\alpha$ is large, behavioral responses to policies are most important at times when there is more social distancing already happening; when $\alpha$ is small this need not be the case.

In Appendix~\ref{sec:calgen}, we repeat the calibration exercise from Section~\ref{sec:cal} with $\alpha=\frac32$ and $\alpha=3$. We find that behavioral responses to changes in risk levels are more important when $\alpha$ is larger, and are more important near the peak of an outbreak for each of the values for $\alpha$.

\subsection{Forward-Looking Agents}\label{sec:fwa}

Our analysis until now has assumed all levels of activity are chosen myopically. We now show that the basic logic behind our results extends to forward-looking agents.

We assume all individuals have discount rate $r>0$. As in the myopic case, an individual pays infection cost $C$ upon transitioning from the susceptible state to the infected state. We choose this timing to maintain consistency between the myopic model and the limit case $r=\infty$ and to obtain clean statements of results, but would obtain similar results under the (likely more realistic) alternate assumption that the infection cost is incurred after infection.

The continuation payoff at time $t_0$ for an individual $i$ choosing activity levels $q_i(t)$ who becomes infected at time $t_1 \geq t_0$ is
$$\int_{t_0}^{\infty}e^{-r(t-t_0)} (q_i(t)-aq_i(t)^2) dt - C e^{-r(t_1-t_0)}.$$
while the continuation payoff at time $t_0$ for an individual $i$ choosing activity levels $q_i(t)$ who does not become infected at any time $t_1 \geq t_0$ is
$$\int_{t_0}^{\infty}e^{-r(t-t_0)} (q_i(t)-aq_i(t)^2) dt.$$
An equilibrium is a sequence of choices of activity levels $q_S(t)$ and $q_R(t)$ and disease states $(S(t), I(t), R(t))$ consistent with the contagion dynamics and such that all individuals maximize their expected continuation payoffs at all times $t$. Individuals therefore respond optimally to the future disease path, which is common knowledge at equilibrium.

Fix an equilibrium and let $\pi_S(t)$, $\pi_I(t)$, and $\pi_R(t)$ be the expected continuation payoffs to individuals in the susceptible, infected, and recovered states, respectively.\footnote{In principle, an equilibrium need not exist and multiple equilibria are possible with forward-looking agents (unlike in the myopic case, where there is a unique equilibrium).  In numerical work on SIR models with endogenous behavior, existence and multiplicity of equilibrium have not been concerns.} The continuation payoffs $\pi_I(t_0)$ and $\pi_R(t_0)$ at time $t_0$ are equal to the conditional expectation of 
$$\int_{t_0}^{\infty}e^{-r(t-t_0)} (q_i(t)-aq_i(t)^2) dt$$
given the disease state at time $t_0$, while $\pi_S(t_0)$ also includes the expected discounted infection cost for an individual who is susceptible at time $t_0$.

We define $\widetilde{C}(t) = C+(\pi_I(t)-\pi_S(t))$ to be the change in payoffs from infection. This is an endogenous object depending on future equilibrium actions and the future likelihood of  infection, as well as the infection cost. The equilibrium activity level at each time $t$ then satisfies
\begin{equation}\label{eq:actfl}q(t)=\frac{1}{2a+\widetilde{C}(t) \beta  \cdot \frac{S(t)I(t)}{S(t)+I(t)}}.\end{equation}
Infection continues to carry a direct cost due to health impacts, but also saves the infected individual from the possibility of future infection. An individual who has been infected and is now immune can interact more without consequence.

We first prove several properties of the change in continuation payoffs from infection $\widetilde{C}(t)$. The bounds we now show can be used to reformulate subsequent results involving $\widetilde{C}(t)$ in terms of the population disease state and exogenous parameters. Let $\overline{R}(\infty)$ be the fraction of the population eventually infected if $q(t)=\overline{q}$ for all $t$, i.e., in the standard SIR model with no behavioral response.
\begin{prop}\label{prop:ctilde}
Suppose the initial prevalence is $I(0)>0$. The continuation payoffs from infection satisfy:

(i) $\left(1-\frac{\overline{R}(\infty) - R(t)}{1-R(t)}\right)C < \widetilde{C}(t) < C$ for all $t$,

(ii) $\lim_{t\rightarrow\infty} \widetilde{C}(t)=C$, and

(iii) $\lim_{I(0)\rightarrow 0} \widetilde{C}(0)=C$.
\end{prop}

Part (i) of the proposition gives bounds on $\widetilde{C}(t)$. The continuation payoffs of infected individuals are always higher than the  continuation payoffs of susceptible individuals, so $\widetilde{C}(t)<C$. The lower bound on $\widetilde{C}(t)$ is derived by comparing the continuation payoffs at equilibrium to the continuation payoffs if there were no social distancing. The lower bound is expressed in terms on the current recovered population $R(t)$ and $\overline{R}(\infty)$, which can be computed in the standard SIR model.

Part (ii) says that late in the outbreak, most infections and social distancing have already occurred and so the direct infection costs dominate. Part (iii) says that the initial prevalence is very small, direct infection costs also dominate at the beginning of the outbreak. More generally, whether $\widetilde{C}(t)$ is large early in the outbreak will depend on the initial prevalence and the discount rate.

We next formulate our risk compensation results with forward looking agents. We study the impact of a temporary shock at time $t$, with the relevant parameter reverting at all times after $t$. For example, an unexpected change in temperature could create a brief temporary shock to the transmission rate. We assume that behavior continues to satisfy equation~(\ref{eq:actfl}) and that the equilibrium selection at times after $t$ is not affected by the shock.

The definition of high infection risk now replaces the infection cost $C$ with the endogenous cost $\widetilde{C}$. 

\begin{defin}
The \textbf{high infection risk} condition holds at time $t$ if $$\widetilde{C}(t) \beta \cdot \frac{S(t)I(t)}{I(t)+S(t)}>2 a.$$
\end{defin}
As in the myopic model, the left-hand side is larger when the susceptible and infected populations are larger. There is now an additional force, as $\widetilde{C}(t)$ is larger when there is less equilibrium social distancing and a lower likelihood of infection in the near future (see Proposition~\ref{prop:ctilde}).

The risk compensation results extend from the myopic case:

\begin{propquad}{thm:prevalenceIR}
Suppose $I(t)< S(t)$. Then $\left. \frac{\partial \iota^x(t)}{\partial x}\right|_{x=0}<0$
if and only if the high infection risk condition holds at time $t$.
\end{propquad}

\begin{propquad}{thm:policy}
The following are equivalent:
\begin{enumerate}[(1)]
\item The high infection risk condition holds at time $t$,
\item Marginally decreasing $\beta$ at time $t$  increases $\iota(t)$, and
\item Marginally decreasing $C$ at time $t$ increases $\widetilde{C}(t)\iota(t)$.
\end{enumerate}
\end{propquad}

The basic logic is the same as in the myopic case. The term $\widetilde{C}(t) \beta \cdot \frac{S(t)I(t)}{I(t)+S(t)}$ now captures the behavioral response to a temporary change. The assumption of an instantaneous shock that does not affect future play ensures that $\widetilde{C}(t)$ is unchanged. A lasting change, such as a more contagious variant, would also change $\widetilde{C}(t)$ by altering the continuation payoffs for susceptible and infected individuals. This would introduce an additional term involving the derivative of $\widetilde{C}(t)$.

The high infection risk condition in the two theorems includes the term $\widetilde{C}(t)$, which depends on the future equilibrium path. The bounds in Proposition~\ref{prop:ctilde} imply a necessary condition for high infection risk in terms of the population disease state and exogenous parameters.

\bibliographystyle{ecta}
\bibliography{Contagious}

\appendix
\section{Proofs of Theorems~\ref{thm:prevalenceIR} and~\ref{thm:policy}}

The proofs of Theorems~\ref{thm:prevalenceIR} and~\ref{thm:policy} are special cases of the proofs of Theorems~\ref{thm:prevalenceIR}$'$ and~\ref{thm:policy}$'$, which we provide here. The remaining proofs are given in Appendix~\ref{sec:remproofs}.

\begin{proof}[Proof of Theorem~\ref{thm:prevalenceIR}$'$]
Let $q_S^x(t)$ be the equilibrium action of susceptible and infected individuals given $(S^x(t),I^x(t),R^x(t))$. As $x$ varies, $S^x(t)+I^x(t)$ and $R^x(t)$ remain constant while $\left. \frac{\partial \left(S^x(t)I^x(t)\right)}{\partial x}\right|_{x=0}=S(t)-I(t)$.

The infection rate is $\iota^x(t) = q_S^x(t)^2S^x(t)I^x(t)\beta$, and we want to compute the derivative of the infection rate in $x$. Differentiating, $$ \left.\frac{\partial \iota^x(t) }{\partial x}\right|_{x=0}= q_S(t)^2 (S(t)-I(t))\beta+ 2q_S(t)\left(\left.\frac{\partial q_S^x(t) }{\partial x}\right|_{x=0}\right) S(t)I(t) \beta.$$
The right-hand side has the same sign as $q_S(t)(S(t)-I(t)) + 2 \left(\left.\frac{\partial q_S^x(t) }{\partial x}\right|_{x=0}\right) S(t)I(t).$

Recall equation~(\ref{eq:payoffcompqs}), which says 
$$q^x_S(t) = \frac{1+\frac{b R^x(t)}{2a-b}}{2a-b(S^x(t)+I^x(t))(1+\frac{b R^x(t)}{2a-b}) + C\beta \cdot \frac{S^x(t)I^x(t)}{S^x(t)+I^x(t)}}.$$
Differentiating this expression in $x$,
$$\left(\left.\frac{\partial q_S^x(t) }{\partial x}\right|_{x=0}\right) =- \frac{ C \beta \cdot  \frac{S(t)-I(t)}{S(t)+I(t)} \cdot \left(1+\frac{b R(t)}{2a-b}\right)}{\left(2a-b(S(t)+I(t))\left(1+\frac{b R(t)}{2a-b}\right) + C\beta \cdot \frac{S(t)I(t)}{S(t)+I(t)}\right)^2}.$$
Substituting shows that $q_S(t)(S(t)-I(t)) + 2 \left(\left.\frac{\partial q_S^x(t) }{\partial x}\right|_{x=0}\right)  S(t)I(t)$ is equal to
\begin{equation}\label{eq:exptosignpc} \frac{(S(t)-I(t))q_S(t)\cdot \left(  2a-b(S(t)+I(t))\left(1+\frac{b R(t)}{2a-b}\right)-C \beta \cdot \frac{S(t)I(t)}{I(t)+S(t)} \right) }{2a-b(S(t)+I(t))\left(1+\frac{b R(t)}{2a-b}\right) + C\beta \cdot \frac{S(t)I(t)}{S(t)+I(t)}} .\end{equation}
Suppose $I(t) < S(t)$. The denominator is positive by equation~(\ref{eq:payoffcompqs}), since $q_S(t) > 0$. So expression~(\ref{eq:exptosignpc}) is negative if and only if $$C \beta \cdot \frac{S(t)I(t)}{I(t)+S(t)}>  2a-b(S(t)+I(t))\left(1+\frac{b R(t)}{2a-b}\right),$$ which is the definition of high infection risk at time $t$.
\end{proof}

\begin{proof}[Proof of Theorem~\ref{thm:policy}$'$]
We will again use our formula for the infection rate $\iota(t) = q_S(t)^2S(t)I(t)\beta$ and the characterization of $q_S(t)$ from equation~(\ref{eq:payoffcompqs}).

(1) $\Leftrightarrow$ (2): Differentiating the infection rate, $$ \frac{\partial \iota(t) }{\partial \beta}= q_S(t)^2 S(t)I(t)+ 2q_S(t)\frac{\partial q_S(t) }{\partial \beta} S(t)I(t) \beta.$$
The right-hand side has the same sign as $q_S(t) + 2 \frac{\partial q_S(t) }{\partial \beta} \beta.$ Differentiating equation~(\ref{eq:focgen}),
$$\frac{\partial q_S(t) }{\partial \beta} = - \frac{ C  \cdot  \frac{S(t)I(t)}{S(t)+I(t)} \cdot \left(1+\frac{b R(t)}{2a-b}\right)}{\left(2a-b(S(t)+I(t))\left(1+\frac{b R(t)}{2a-b}\right) + C\beta \cdot \frac{S(t)I(t)}{S(t)+I(t)}\right)^2}.$$
Substituting, 
\begin{equation}\label{eq:exptosignpcpolicy}q_S(t) + 2 \frac{\partial q_S(t) }{\partial \beta} \beta = \frac{q_S(t)\left(  2a-b(S(t)+I(t))\left(1+\frac{b R(t)}{2a-b}\right)-C \beta \cdot \frac{S(t)I(t)}{I(t)+S(t)} \right) }{2a-b(S(t)+I(t))\left(1+\frac{b R(t)}{2a-b}\right) + C\beta \cdot \frac{S(t)I(t)}{S(t)+I(t)}} .\end{equation}
The denominator is positive by equation~(\ref{eq:payoffcompqs}), since $q_S(t) > 0$. So expression~(\ref{eq:exptosignpcpolicy}) is negative if and only if $$C \beta \cdot \frac{S(t)I(t)}{I(t)+S(t)}>  2a-b(S(t)+I(t))\left(1+\frac{b R(t)}{2a-b}\right),$$ which is the definition of high infection risk at time $t$.

(2) $\Leftrightarrow$ (3): Because $C$ is constant, the infection rate $\iota(t)$ is decreasing in $\beta$ if and only if $C\iota(t)$ is decreasing in $\beta$. In our expression for $C\iota(t)$, the terms $C$ and $\beta$ only appear within the product $C\beta$. So $C\iota(t)$ is decreasing in $\beta$ if and only if it is decreasing in $C$.
\end{proof}

\newpage

\setcounter{page}{1}

\begin{center} {\Large Online Appendix} \end{center}

\section{Basic Results}\label{sec:basic}

We begin with a basic result about the long-run population.

\begin{prop}\label{prop:limits}
Suppose that $S(0)>0$ and $I(0)>0$. Then $\lim_{t\rightarrow \infty} I(t) = 0$ and $\lim_{t\rightarrow \infty} R(t)\in (0,1)$.
\end{prop}
The proposition says that the disease prevalence eventually approaches zero and that some individuals will never be infected. While contagion vanishes due to herd immunity, the eventual fraction of the population who has had the disease depends on behavior and any policy interventions. We define $R(\infty) = \lim_{t\rightarrow \infty} R(t)$ to be the fraction of the population that is eventually infected.

Examples show that  behavioral responses can have a large effect on $R(\infty)$. To see this, consider an example with reproductive number $R_0 = 2.5$ and average disease length of $14$ days (see Section~\ref{sec:cal} for details on the choice of parameter values). The eventual fraction infected $R(\infty)$ is $88\%$ with $C=0$ (no behavioral response) but drops to $64\%$ when the health cost $C$ from infection is equal to the value of three months of normal activity and to $61\%$ with $C$ equal to the value of six months of normal activity (Figure~\ref{fig:Cdiff}). Thus, the model with endogenous behavior predicts a substantially smaller share of the population will eventually be  infected than the standard SIR model would suggest.

\begin{figure}
\center{\includegraphics[scale=.55]{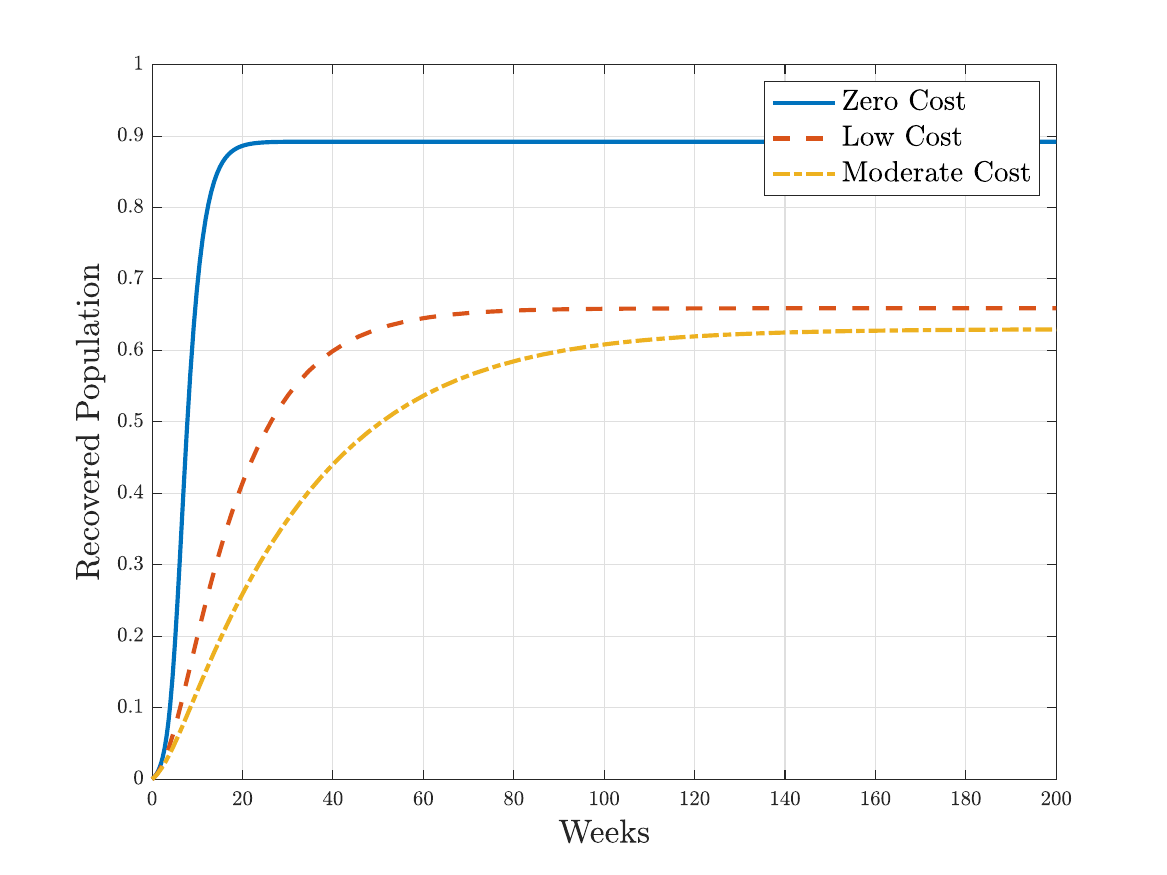}}

\caption{Recovered population over time for infection costs $C$ equal to zero (solid blue), equal to the value of three months of normal activity (dashed red), and equal to the value of six months of normal activity (dash-dotted yellow). The reproductive number is $R_0=2.5$ and the average disease length is $14$ days.}\label{fig:Cdiff}
\end{figure}

Behavioral responses also matter for the speed of convergence since social distancing slows down herd immunity effects. In Figure~\ref{fig:Cdiff}, almost all infections occur within three months with zero infection costs while almost all infections occur within two years with moderate infection costs.

We will discuss how shocks or policy changes affect $R(\infty)$ below. Before doing so, we describe the disease path.


\begin{prop}\label{prop:singlepeaked}
For initial disease prevalence $I(0)>0$ sufficiently small, if $I(t) \leq S(t)$ for all $t< t^*$ where $t^* \leq \infty$, then the prevalence $I(t)$ is single-peaked on the interval $[0,t^*]$.
\end{prop}

\begin{figure}
\center{\includegraphics[scale=.7]{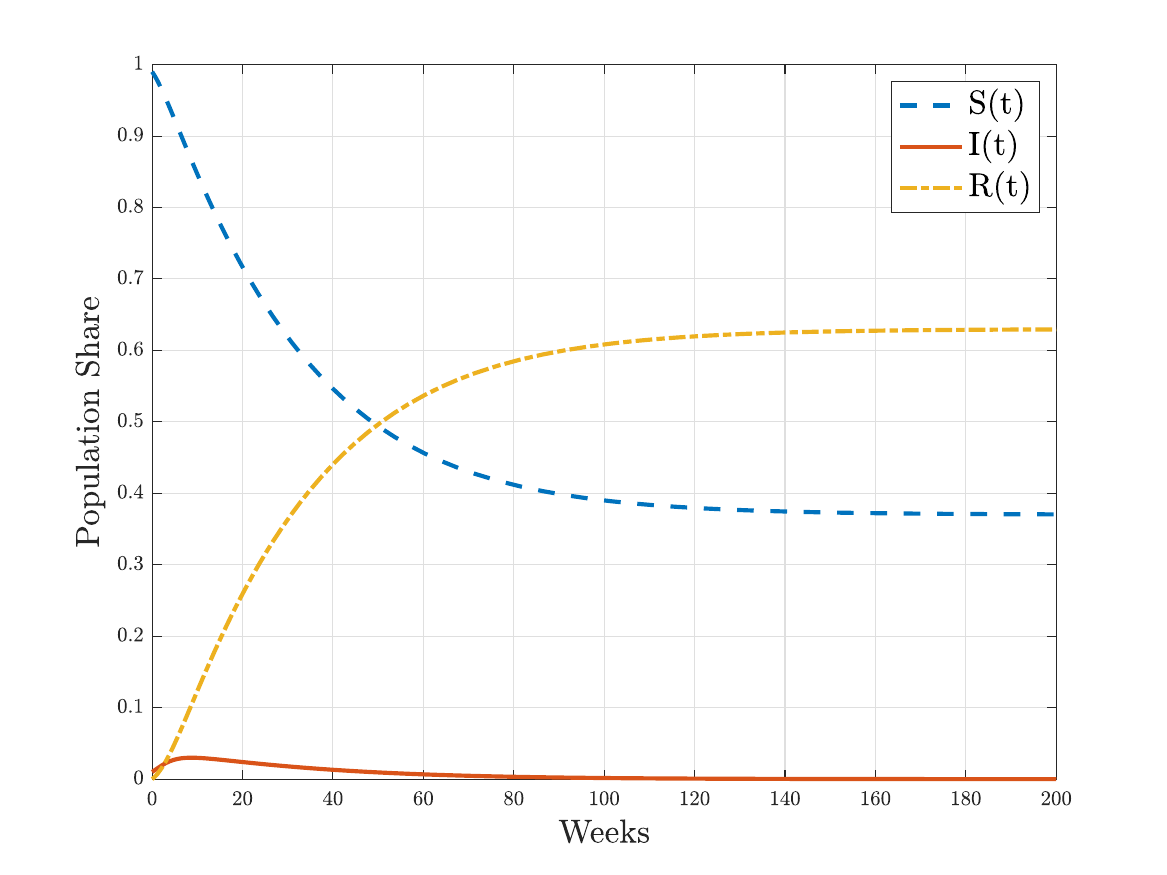}

\includegraphics[scale=.7]{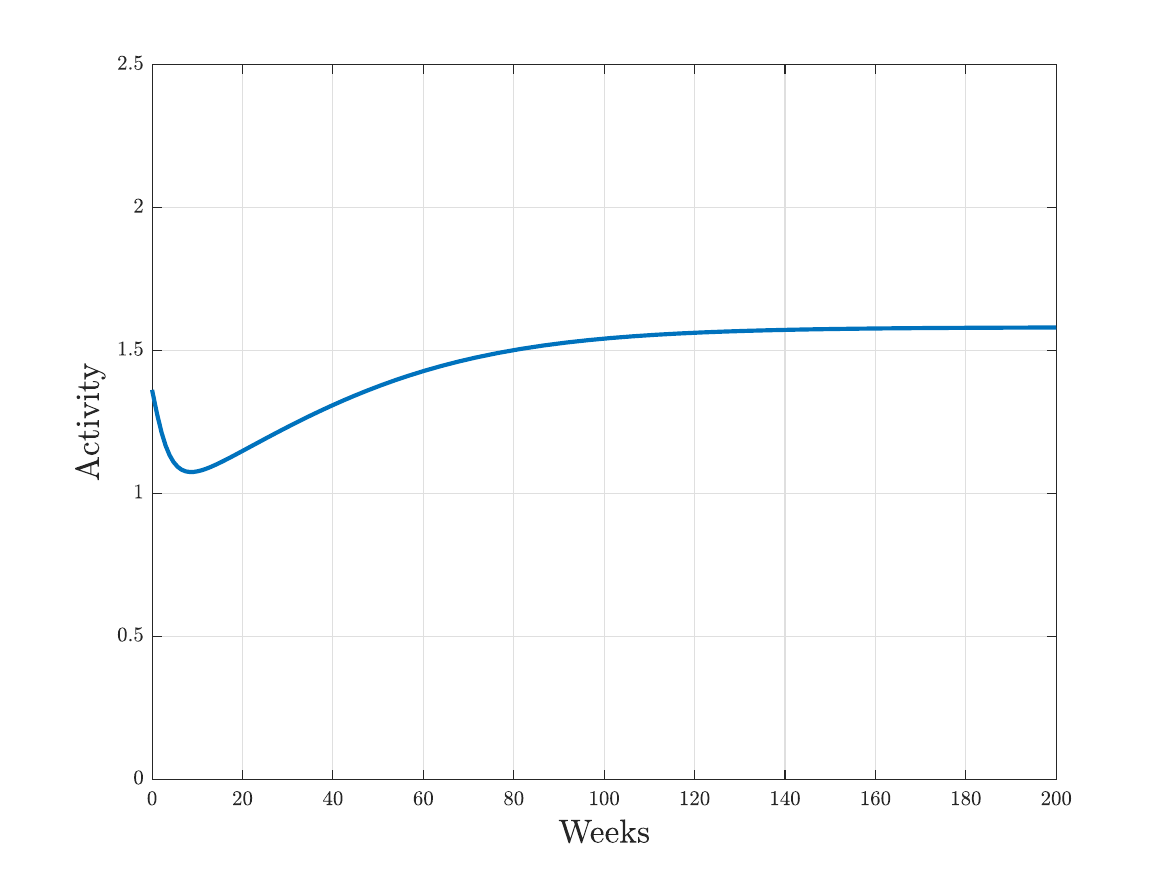}}

\caption{Basic dynamics of the population and behavior with $I(t)$ single-peaked. Graphs show the time paths of $S(t),$ $I(t)$, and $R(t)$ (top) and $q(t)$ (bottom), with $R_0 = 2.5$, average disease length $14$ days, and infection cost equal to the value of six months of normal activity.}\label{fig:basic}
\end{figure}

The proposition says that as long as the infected population is not larger than the susceptible population, the prevalence of the disease is single peaked. Figure~\ref{fig:basic} shows the path of the disease and behavior in a simple example.

If the population infected $I(t)$ is much larger than the susceptible population $S(t)$, then the prevalence may increase again late in the epidemic. The intuition is that in this case, most individuals who have not yet recovered are unlikely to be susceptible. This mitigates infection risk and leads to more interaction. The proof gives a necessary condition for the increased interaction to increase disease prevalence (see Section~\ref{sec:risk} for similar effects); the condition requires $I(t) > S(t)$ but is actually substantially stronger. This effect seems unlikely to occur in practice, as it would require unusually large prevalences for a severe disease.

There are several more plausible explanations for multiple waves in our model. First, temporary lockdowns which restrict actions $q(t)$ to be below some exogenous level will generate a second wave if they are sufficiently strict and lifted sufficiently early. Second, if people respond to the population disease state at a lag (i.e., at time $t$ all individuals act as if the current disease state is $(S(t-\tau),I(t-\tau),R(t-\tau))$ for some $\tau>0$), then the prevalence need not be single-peaked.

The condition $I(t) \leq S(t)$ in Proposition~\ref{prop:singlepeaked} depends on the endogenous parameters $I(t)$ and $S(t)$. We can also state the single-peakedness result in terms of exogenous parameters, such as the disease cost and transmission rate.
\begin{cor}\label{cor:singlepeaked}
There exist constants $0<\underline{\delta} < \overline{\delta}$ such that if $C\beta<\underline{\delta}$ or $C\beta>\overline{\delta}$, then the prevalence $I(t)$ is single-peaked for $I(0)>0$ sufficiently small.
\end{cor}

The single-peaked prevalence result holds for $C\beta$ sufficiently small or sufficiently large. The threshold levels $\underline{\delta}$ and $\overline{\delta}$ for $C\beta$ may depend on $a$ and $\kappa$. We could also have formulated the corollary for other parameters, i.e., the prevalence is single-peaked when $a$ or $\kappa$ are sufficiently large.

\section{Calibrations with General Payoffs}\label{sec:calgen}

In this section, we repeat the calibration exercise from Section~\ref{sec:cal} with more general payoffs. We take the flow payoffs from economic activity to be $q-aq^{\alpha}$ for $\alpha>1$. Section~\ref{sec:gen} showed that the high infection risk condition is then $$ C\beta \cdot \frac{I(t)S(t)}{S(t)+I(t)}>a\alpha(\alpha-1)q^{\alpha-2}.$$

Figures~\ref{fig:lowalpha} and~\ref{fig:highalpha} conduct simulations for $\alpha=1.5$ and $\alpha=3$. As in Section~\ref{sec:cal}, we consider baseline reproductive numbers $R_0^{low}=2.5$ and $R_0^{high}=7.0$ with average contagious period of $14$ days and infection cost equal to the value of six months of normal activity. Each figure shows the prevalence $I(t)$ in dashed  blue and $ C\beta \cdot \frac{I(t)S(t)}{S(t)+I(t)}-a\alpha(\alpha-1)q^{\alpha-2}$ in solid red. The high infection risk condition holds when the solid red curve is positive.

\begin{figure}[!tbp]
\centering
\begin{minipage}[b]{0.4\textwidth}
\includegraphics[scale=.4]{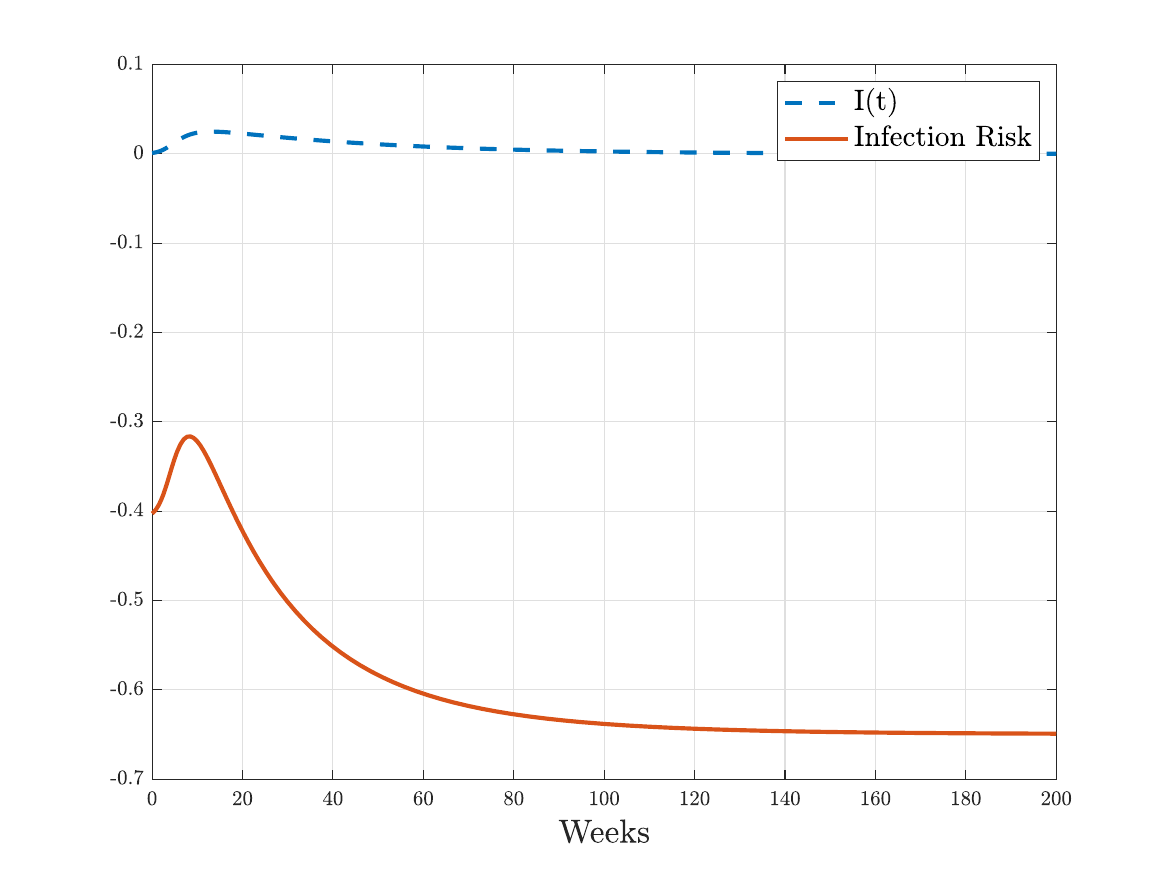}
\end{minipage}
\hfill
\begin{minipage}[b]{0.4\textwidth}
\includegraphics[scale=.4]{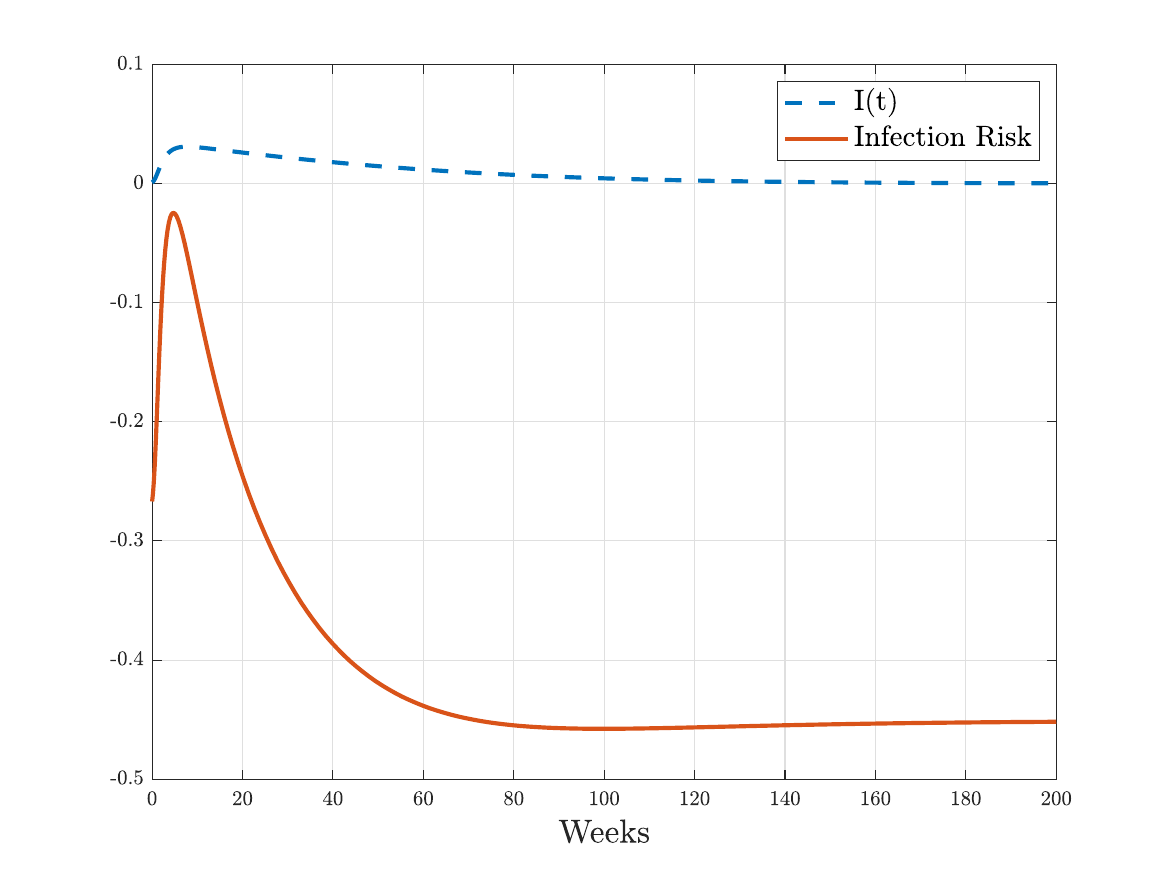}
\end{minipage}
\caption{Prevalence and infection risk with economic payoffs $q-aq^{\frac32}$. Time path of $I(t)$ in dashed blue and $\frac34 a q^{-\frac12}-C \beta \cdot \frac{S(t)I(t)}{I(t)+S(t)}$ in solid red for $R_0^{low}=2.5$ (left) and $R_0^{high}=7.0$ (right) with flow payoffs $q-aq^{3/2}$. The high infection risk condition holds when the value of the solid red curve is positive, which does not occur.}\label{fig:lowalpha}
\end{figure}

\begin{figure}[!tbp]
\centering
\begin{minipage}[b]{0.4\textwidth}
\includegraphics[scale=.4]{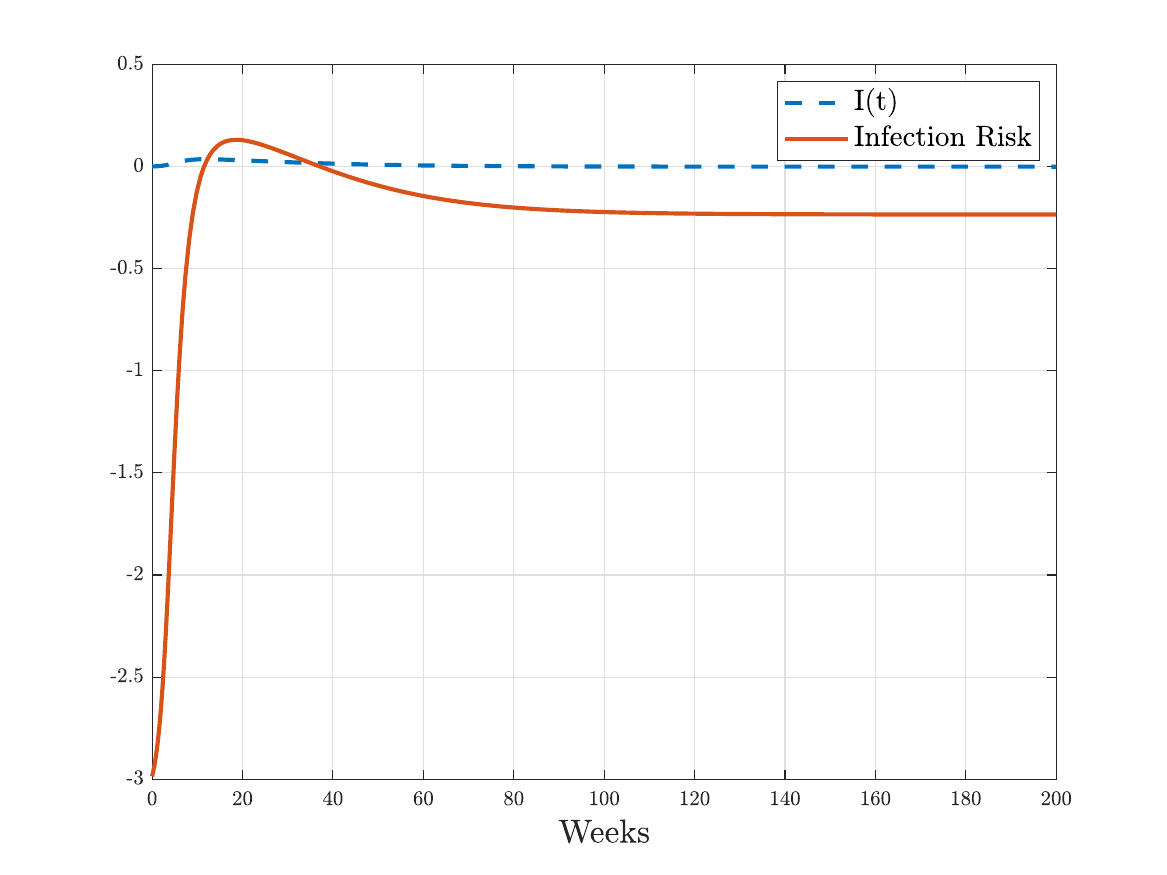}
\end{minipage}
\hfill
\begin{minipage}[b]{0.4\textwidth}
\includegraphics[scale=.4]{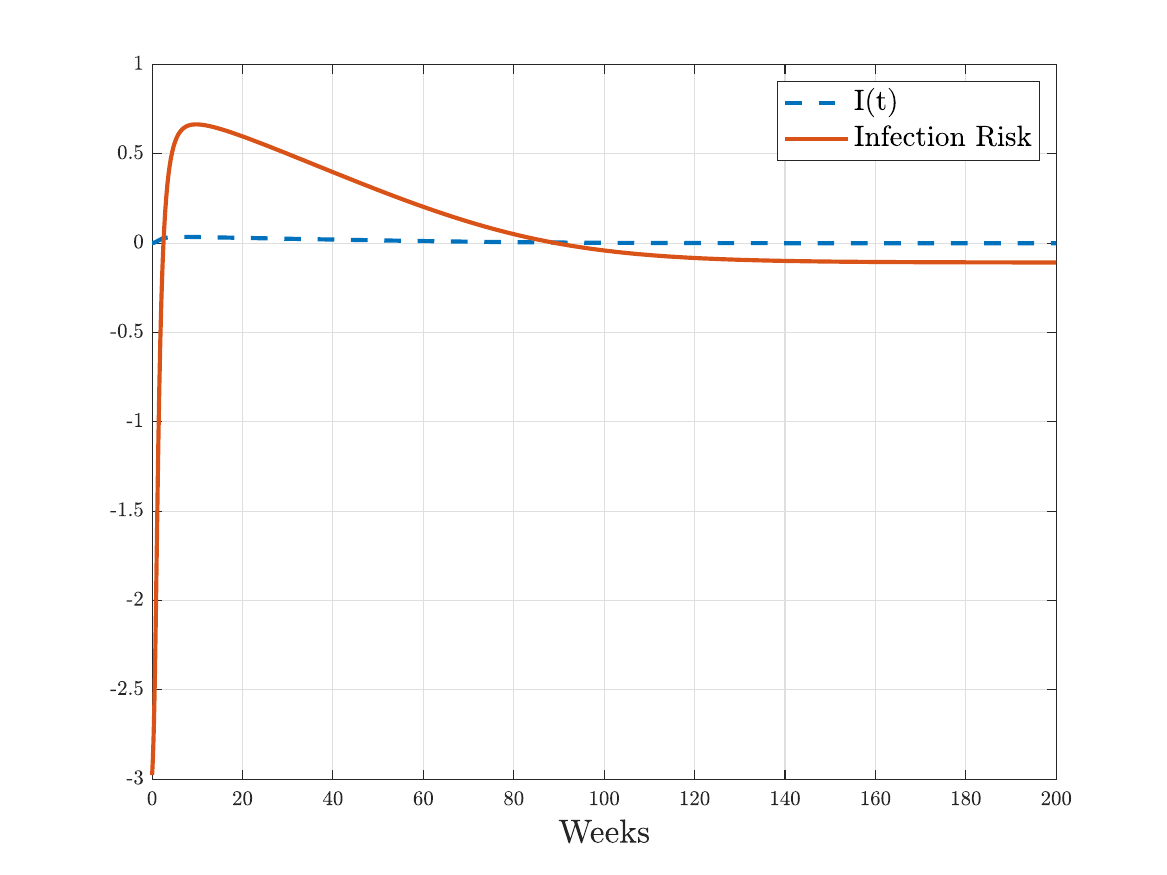}

\end{minipage}
\caption{Prevalence and infection risk with economic payoffs $q-aq^{3}$. Time path of $I(t)$ in dashed blue and $6aq-C \beta \cdot \frac{S(t)I(t)}{I(t)+S(t)}$ in solid red for $R_0^{low}=2.5$ (left) and $R_0^{high}=7.0$ (right) with flow payoffs $q-aq^{3}$. The high infection risk condition holds risk when the value of the solid red curve is positive.}\label{fig:highalpha}
\end{figure}

A higher value of $\alpha$ leads to stronger behavioral responses. When $\alpha=\frac{3}{2}$, the high infection risk condition does not hold for $R_0^{mod}$ or $R_0^{high}$, so the direct effects of policy changes will be larger than behavioral responses. The condition does come close to holding near the peak with reproductive number $R_0^{high}$, however.

When $\alpha=3$, the the high infection risk condition does hold for both $R_0^{mod}$ and $R_0^{high}$. With reproductive number $R_0^{mod}$, the solid red curve's maximum is about $0.1$, suggesting that behavioral responses to policy changes are barely larger than direct effects. With reproductive number $R_0^{mod}$, behavioral responses are substantially larger than direct effects for more than a year.

In both cases, behavioral responses are largest relative to direct effects near the peak of the infection. When $\alpha=\frac32$, this relationship is not entirely monotonic: after falling after the infection peak, the solid red curve $\frac34 a q^{-\frac12}-C \beta \cdot \frac{S(t)I(t)}{I(t)+S(t)}$ actually increases very slightly for $t$ large.

\section{Remaining Proofs}\label{sec:remproofs}

\begin{proof}[Proof of Proposition~\ref{prop:limits}]
We first show that $I(t) \rightarrow 0$. If not, we could choose $\epsilon>0$ and a sequence $t_1,t_2,\hdots$ such that $t_{i+1}> t_i+1$ and $I(t_i) > \epsilon$. Because infected individuals recover at a Poisson rate $\kappa$, there exists $\delta>0$ such that if $I(t) > \epsilon$, at least $\delta$ individuals recover in the time interval $[t,t+1]$. Thus at least $\delta$ individuals recover in period $[t_i,t_i+1]$ for each $i$, which gives a contradiction since there is a unit mass of individuals who can recover at most once each.

Since $R(t)$ is increasing, we can define $R(\infty) = \lim_{t\rightarrow \infty} R(t)$. We next show that this limit $R(\infty ) \in (0,1)$. Because $I(0)>0$ and infected individuals recover at a Poisson rate, $R(\infty)>0$.

Suppose that $R(\infty)=1$, so that $S(t) \rightarrow 0$. Because $q(t) \leq \overline{q}$ for all $t$, we can bound $R(\infty)$ above by the limit of $R(t)$ when all individuals choose activity $q(t)=\overline{q}$ for all $t$:

\begin{lemma}\label{lem:bound}
Let $\overline{R}(t)$ be the share of recovered individuals when all individuals choose level of activity $\overline{q}$ for all $t$. Then $R(\infty) \leq \overline{R}(\infty)$.
\end{lemma}
\begin{proof}
Define an infection chain $i_0,\hdots,i_k$ of length $k$ to be a sequence of $k$ individuals such that $i_0$ is infected at time zero and for each $0\leq j <k$, individual $i_j$ meets individual $i_{j+1}$ while individual $i_j$ is infected. Then $R(\infty)$ is equal to the share of invidiuals who are in an infection chain of length $k$ for any $k$. So it is sufficient to show that the number of individuals in infection chain of length $k$ for any $k$ is weakly higher under level of activity $\overline{q}$ at all times compared to $q(t)$.

We claim that for each $k$, the share of individuals in infection chains of length at most $k$ is weakly higher under actions $\overline{q}$ at all times compared to $q(t)$. The proof of the claim is by induction, and when $k=0$ the share is $I(0)$ in both cases.

Suppose the claim holds for $k$. An individual $i$ who is not initially infected is in an infection chain of length $k+1$ if and only if $i$ meets at least one individual in an infection chain of length at most $k$ while that individual is infected. By the inductive hypothesis, there are weakly more individuals in an infection chain of length at most $k$ under actions $\overline{q}$ at all times compared to $q(t)$. As meetings are independent, the probability of meeting any such individual is increasing in $q_i(t)$ and $q_I(t)$ for all $t$. So the probability that $i$ is in an infection chain of length $k+1$ is weakly higher under actions $\overline{q}$ at all times compared to $q(t)$. Therefore, the claim holds for $k+1$. This completes the induction and therefore proves the lemma.
\end{proof}

By the lemma, we can assume for the remainder of the argument that $q(t)=\overline{q}$ for all $t$. We then have infection rate $S(t)I(t) \overline{q}^2\beta$ and recovery rate $\kappa I(t)$.

We claim that $\frac{S(t)}{I(t)}\rightarrow \infty$. Because $S(t) \rightarrow 0$ and $I(t)\rightarrow 0$, we can choose at time $t_0$ such that $\frac{S(t)\overline{q}^2\beta}{\kappa} \leq \frac14$ and $\frac{I(t)\overline{q}^2\beta}{\kappa} \leq \frac14$ for all $t \geq t_0$. Then for all $t \geq t_0$, we have $\dot{S}(t) \geq -\frac{\kappa}{4} S(t)$ by the first inequality and $\dot{I}(t) \leq -\frac{\kappa}{2} I(t)$ by the second inequality. In particular, $I(t)$ exponentially decays at a faster rate than $S(t)$. Since $S(t_0)$ and $I(t_0)$ are positive, this proves the claim.

Therefore, we can choose $t_0$ such that $S(t)>I(t)$ for all $t \geq t_0$ and $S(t)\overline{q}^2 \beta<\frac{\kappa}{2}$ for all $t \geq t_0$. Given an individual $i$ who is infected at time $t \geq t_0$, the second inequality implies that the expected number of susceptible individuals whom $i$ infects is less than $\frac12$. Therefore, the total number of individuals who become infected at times $t \geq t_0$ is less than $$I(t_0)\sum_{k=1}^{\infty} \frac{1}{2^k} =I(t_0) < S(t_0).$$
This contradicts our assumption that $S(t) \rightarrow 0$, so we must have $R(\infty) < 1$.
\end{proof}

\begin{proof}[Proof of Proposition~\ref{prop:singlepeaked}]
We first show that $I(t)$ is initially increasing for $I(0)$ small. We then show that once $\dot{I}(t)=0$, we have $\dot{I}(t) \leq 0$ at all subsequent times until $t^*$. To do so we show that whenever $\dot{I}(t)=0$ at some time $t < t^*$, we have $\ddot{I}(t)<0$ at this time. This step uses the assumption $S(t)\geq I(t)$.

Note that the population shares $(S(t),I(t),R(t))$ and the equilibrium level of activity $q(t)$ are continuous in time, so $\dot{I}(t)$ is continuous as well.

We first show that for $I(0)$ sufficiently small, $\dot{I}(0)>0$. By Assumption \ref{assumption} (high enough transmissivity), $\overline{q}^2 \beta>\kappa$. Because $q(0)$ is a continuous function of $(S(t),I(t),R(t))$ with $q(0)=\overline{q}$ when $I(t)=0$, it follows that $$(1-I(0))q(0)^2 \beta > \kappa$$
for $I(0)>0$ sufficiently small. Since $\dot{I}(0)=I(0)(1-I(0))q(0)^2 \beta -I(0)\kappa$, the prevalence is initially increasing for $I(0)$ sufficiently small.

Let $t^{peak}$ be the first time at which $\dot{I}(t)=0$.\footnote{If there is no such $t < t^*$, then $t^*$ is the peak.} We will show that $\dot{I}(t) < 0$ on $(t^{peak},t^*)$ since we have assumed that $S(t)>I(t)$ on $(t^{peak},t^*)$.

We first claim that $\ddot{I}(t^{peak})<0$. Let $$(S^x(t),I^x(t),R^x(t))=(S(t)-x,I(t),R(t)+x)$$
and let $q^x(t)$ be the corresponding equilibrium action. The infection rate is $\iota^x(t) = q^x(t)^2S^x(t)I^x(t)\beta$, and from equation (\ref{eq:activity}) $$q^x(t)=\frac{1}{2a+C\beta \cdot  \frac{I^x(t)S^x(t)}{S^x(t)+I^x(t)}} \Rightarrow \left. \frac{\partial q^x(t)}{\partial x}\right|_{x=0}= \frac{1}{(2a+C\beta \cdot  \frac{I(t)S(t)}{S(t)+I(t)})^2}\cdot \frac{I(t)^2}{(S(t)+I(t))^2}.$$ So the derivative of the infection rate in $x$ at zero is
\begin{align*}\left. \frac{\partial \iota^x(t)}{\partial x}\right|_{x=0} &= q(t)\beta I(t) \left(-q(t)+2\left.\frac{\partial q^x(t)}{\partial x}\right|_{x=0}S(t)\right)
\\&=- \frac{q (t)\beta I(t)}{(2a+C \beta \cdot \frac{S(t)I(t)}{S(t)+I(t)})^2}\left(\left(2a + C \beta \cdot \frac{S(t)I(t)(S(t)+I(t))}{(S(t)+I(t))^2}\right)-2C\beta \cdot \frac{S(t)I(t)^2}{(S(t)+I(t))^2}\right)
\\&= -\frac{q (t)\beta I(t)}{(2a+C \beta \cdot \frac{S(t)I(t)}{S(t)+I(t)})^2}\left(2a + C \beta \cdot \frac{S(t)}{(S(t)+I(t))^2}\cdot (S(t)-I(t))\right).\end{align*}
We will need that the right-hand side is negative, which holds whenever $S(t)>I(t)$.

At the peak, we have $\dot{I}(t^{peak})=0$ while $\dot{S}(t^{peak})=-\dot{R}(t^{peak})=-\iota(t^{peak}).$ Since we have shown $\left. \frac{\partial \iota^x(t)}{\partial x}\right|_{x=0}<0$, this implies that the derivative of the infection rate $\dot{\iota}(t^{peak})<0$. Because the time derivative of the recovery rate is $\kappa\dot{I}(t^{peak})=0$, the claim $\ddot{I}(t^{peak})<0$ holds.

Similarly, suppose that $\dot{I}(t)=0$ for some $t \in (t^{peak},t^*)$. The same argument shows that $\ddot{I}(t)<0$ at time $t$. So $I(t)$ is strictly decreasing on $[t^{peak}, t^*]$.
\end{proof}

\begin{proof}[Proof of Corollary~\ref{cor:singlepeaked}]
We treat the two cases of $C\beta$ small and $C\beta$ large separately.

\textbf{$C\beta$ small}: By the proof of Proposition~\ref{prop:singlepeaked}, $I(t)$ is single-peaked if $I(0)$ is sufficiently small and $$2a + C \beta \cdot \frac{S(t)}{(S(t)+I(t))^2}\cdot (S(t)-I(t))>0$$
for all $t$.

We claim the expression $-\frac{S(t)I(t)}{(S(t)+I(t))^2}$ is bounded below by a constant $M$ independent of $C$ and $\beta$ (but potentially depending on $a$ and $\kappa$). Indeed, $S(t)I(t)$ is bounded above by one. The expression $S(t)+I(t)$ is bounded below by the value of $1-R(\infty)$ when $\beta=1$ and $C=0$, so that $q(t)=\overline{q}$ for all $t$ (Lemma~\ref{lem:bound}). Since $1-R(\infty)>0$ by Proposition~\ref{prop:limits}, this proves the claim.

We conclude that whenever $C\beta < \frac{2a}{M},$ the prevalence $I(t)$ is single-peaked for $I(0)$ sufficiently small.

\textbf{$C\beta$ large}:  We next consider $C\beta$ large, and will show that for any $C\beta$ sufficiently large, we have $I(t)<S(t)$ for all $t$ when $I(0)$ is small enough. As above, we can bound $1-R(\infty)$ below by a positive constant independent of $C$ and $\beta$, which we call $N$.

Let $\delta<N/2$. We will show that for $C\beta$ sufficiently large and $I(0)<\delta/2$, we have $I(t) < \delta$ for all $t$. If not, there must exist some time $t_0$ at which $\delta/2<I(t_0)<\delta$ and $\dot{I}(t_0)>0$.

Since $I(t_0)+S(t_0)\geq N$, $I(t_0)<\delta$, and $\delta<N/2$, we have $S(t_0) >I(t_0)$. Therefore, $\frac{I(t_0)S(t_0)}{S(t_0)+I(t_0)}\geq \frac{I(t_0)}{2} \geq \frac{\delta}{4}$. Recall that $q(t)=\frac{1}{2a+C\beta\cdot\frac{I(t)S(t)}{S(t)+I(t)}}$. So taking $C \beta$ sufficiently large, we can assume that $q(t_0)<\sqrt{\kappa}$.  Then $$\dot{I}(t_0) = I(t_0) \left(\frac{S(t_0)}{I(t_0)+S(t_0)} q(t_0)^2\beta - \kappa\right) \leq  I(t_0)( q(t_0)^2-\kappa)< 0.$$ But this contradicts our assumption that $\dot{I}(t_0)>0$.

We have shown $I(t)<\delta$ for all $t$. Since  $I(t)+S(t)\geq N > 2\delta$ for all $t$, we must have $I(t)<S(t)$ for all $t$ as desired. We now apply Proposition~\ref{prop:singlepeaked}.
\end{proof}

\begin{proof}[Proof of Proposition~\ref{prop:het}]
Suppose types in $A$ choose $q_k(t)>0$ at equilibrium while the remaining types choose $q_k(t)=0$. For each type $k$ who chooses $q_k(t)>0$, we have a first-order condition
$$2aq_k(t) + C_k \left(\sum_{k'\in A} q_{k'}(t) I_{k'}(t)\right) \beta \frac{S_k(t)}{S_k(t)+I_k(t)}=1.$$
Recalling that $M_A = \left(C_k\beta \gamma_{kk'}\cdot\frac{S_k(t)I_{k'}(t)}{I_k(t)+S_k(t)}\right)_{ k,k' \in A}$, we can rewrite the previous equation as
\begin{equation}\label{eq:hetchar}\left(M_A + 2 \cdot \text{diag}(a_k) \right) (q_k(t))_{k \in A}=  \textbf{1}.\end{equation}
This gives the desired expression for equilibrium behavior.

To complete the proof, we must show that $A=A_m$. At least one entry of $\left(M_A + 2 a I \right)^{-1} \textbf{1}$ must be positive for each $A$. Therefore, at most $m-1$ types can be removed from some $A_j$ with $1\leq j\leq m-1$, so we must have $A_j=A_{j+1}$ beginning with some $1 \leq j \leq m$.

The solution to equation~(\ref{eq:hetchar}) for each $A$ can be seen as the equilibrium of a game for types in $A$ in which levels of activity can be negative and payoffs for each type $k$ are $$q_k(t) - a_kq_k(t)^2-C_k\beta q_k(t) \cdot \frac{S_k(t)}{S_k(t)+I_k(t)} \cdot\sum_{k' \in A} q_{k'}(t) I_{k'}(t).$$ The best response for any individual $i$ at time $t$ is decreasing in $\sum_{k =1}^m I_k(t) q_k(t)$. Each time we pass from $A_j$ to $A_{j+1}$, we increase the actions $q_k(t)$ of all types choosing negative levels of activity according to equation~(\ref{eq:hetchar}) to zero. When the actions of these types are increased to zero, other types will decrease their activity levels. However $\sum_{k =1}^m I_k(t) q_k(t)$ cannot be lower at equilibrium under $A_{j+1}$ than $A_j$ (because then all types would choose higher activity levels under $A_{j+1}$ than $A_j$, which would increase $\sum_{k =1}^m I_k(t) q_k(t)$). Therefore any type choosing zero activity level under $A_j$ for $j<m$ will also choose zero activity level under $A_m$.

The same logic shows that the equilibrium is unique: because the best responses of all types are decreasing in $\sum_{k =1}^m I_k(t) q_k(t)$, there is a unique fixed point.
\end{proof}

\begin{proof}[Proof of Proposition~\ref{prop:twotype}]
We first prove the characterization of when individuals of type $2$ choose $q_2(t)=0.$ This occurs if and only if
$$C_2 \beta q_1(t) I_1(t) \cdot \frac{S_2(t)}{S_2(t)+I_2(t)} \geq 1.$$
If $q_2(t)=0$, then $q_1(t)$ is determined by equation (\ref{eq:activity}) as in the homogeneous case, i.e. $q_1(t) = \frac{1}{2a+C_1 \beta \cdot \frac{I_1(t)S_1(t)}{I_1(t)+S_1(t)}}$. Substituting gives the expression in the proposition.

Next, we show that $I_1(t_1)>I_1(t_2)$ whenever type $2$ stops choosing positive activity at time $t_1$ and resumes at time $t_2$. Suppose not, so that $I_1(t_1)\leq I_1(t_2)$.

We have $$C_2 \beta q_1(t) I_1(t) \cdot \frac{S_2(t)}{S_2(t)+I_2(t)} = 1$$
at times $t_1$ and $t_2$. We must have $\frac{S_2(t_1)}{S_2(t_1)+I_2(t_1)} < \frac{S_2(t_2)}{S_2(t_2)+I_2(t_2)}$, as a positive fraction of infected type $2$ individuals recover while no new type $2$ individuals are infected during the interval $[t_1,t_2]$. Therefore, we must have $q_1(t_1)I_1(t_1) > q_1(t_2)I_1(t_2).$ Since $I_1(t_1)\leq I_1(t_2)$ while $S_1(t_1) > S_1(t_2)$, we have $\frac{S_1(t_1)}{I_1(t_1)+S_1(t_1)} > \frac{S_1(t_2)}{I_1(t_2)+S_1(t_2)}$. Since $q_1(t)I_1(t)$ decreases and the probability of a type $1$ individual being susceptible decreases from $t_1$ to $t_2$, we have $q_1(t_1) < q_1(t_2)$.

We have shown that if $I_1(t_1)\leq I_1(t_2)$, then $q_1(t_1) < q_1(t_2)$ but $q_1(t_1)I_1(t_1) > q_1(t_2)I_1(t_2).$ This gives a contradiction.
\end{proof} 

\begin{proof}[Proof of Theorem~\ref{thm:prevalenceIR}$''$]
Let $q^x(t)$ be the equilibrium action given $(S^x(t),I^x(t),R^x(t))$. As $x$ varies, $S^x(t)+I^x(t)$ remains constant while $\left.\frac{\partial \left(S^x(t)I^x(t)\right)}{\partial x}\right|_{x=0}=S(t)-I(t)$.

The infection rate is $\iota^x(t) = q^x(t)^2S^x(t)I^x(t)\beta$, and we want to compute the derivative of the infection rate in $x$. Differentiating, $$\left. \frac{\partial \iota^x(t) }{\partial x}\right|_{x=0}= q(t)^2 (S(t)-I(t))\beta+ 2q(t)\left(\left. \frac{\partial q^x(t)}{\partial x}\right|_{x=0}\right) S(t)I(t) \beta.$$
The right-hand side has the same sign as $q(t)(S(t)-I(t)) + 2 \left(\left. \frac{\partial q^x(t)}{\partial x}\right|_{x=0}\right) S(t)I(t).$

Equation~(\ref{eq:focgen}), which applies at equilibrium since $q(t)>0$, is now:
$$f'\left(q^x(t)\right) - C \beta q^x(t)  \cdot \frac{I^x(t)S^x(t)}{S^x(t)+I^x(t)} = 0.$$ 
Differentiating this expression in $x$,
$$\left(\left. \frac{\partial q^x(t)}{\partial x}\right|_{x=0}\right) = \frac{ C \beta q(t)\cdot  \frac{S(t)-I(t)}{S(t)+I(t)}}{f''(q(t))-C\beta \cdot  \frac{I(t)S(t)}{S(t)+I(t)}}.$$
Substituting shows that $q(t)(S(t)-I(t)) + 2 \left(\left.\frac{\partial q^x(t) }{\partial x}\right|_{x=0}\right)  S(t)I(t)$ is equal to
\begin{equation}\label{eq:exptosign} \frac{(S(t)-I(t))q(t)}{f''(q(t))-C\beta \cdot  \frac{I(t)S(t)}{S(t)+I(t)}}  \cdot \left(  C \beta \cdot \frac{S(t)I(t)}{I(t)+S(t)}+f''(q(t)) \right) .\end{equation}
Suppose $I(t) < S(t)$. The denominator is negative because $f$ is concave. So expression~(\ref{eq:exptosign}) is negative if and only if $$C \beta \cdot \frac{S(t)I(t)}{I(t)+S(t)}> -f''(q(t)),$$ which is the definition of high infection risk at $t$.
\end{proof}

\begin{proof}[Proof of Theorem~\ref{thm:policy}$''$]
We will again use our formula for the infection rate $\iota(t) = q(t)^2S(t)I(t)\beta$ and the characterization of $q(t)$ from equation~(\ref{eq:focgen}), which applies at equilibrium since $q(t)>0$.

(1) $\Leftrightarrow$ (2): Differentiating the infection rate, $$ \frac{\partial \iota(t) }{\partial \beta}= q(t)^2 S(t)I(t)+ 2q(t)\frac{\partial q(t) }{\partial \beta} S(t)I(t) \beta.$$
The right-hand side has the same sign as $q(t) + 2 \frac{\partial q(t) }{\partial \beta} \beta.$ Differentiating equation~(\ref{eq:focgen}),
$$\frac{\partial q(t) }{\partial \beta} = \frac{ C  q(t)\cdot  \frac{I(t)S(t)}{S(t)+I(t)}}{f''(q(t))-C\beta \cdot  \frac{I(t)S(t)}{S(t)+I(t)}}.$$
Substituting, 
\begin{equation}\label{eq:exptosignpolicy}q(t) + 2 \frac{\partial q(t) }{\partial \beta} \beta = \frac{q(t)}{f''(q(t))-C\beta \cdot  \frac{I(t)S(t)}{S(t)+I(t)}}  \cdot \left(  C \beta \cdot \frac{S(t)I(t)}{I(t)+S(t)}+f''(q(t)) \right) .\end{equation}
The denominator is negative because $f$ is concave. So expression~(\ref{eq:exptosignpolicy}) is negative if and only if $$C \beta \cdot \frac{S(t)I(t)}{I(t)+S(t)}> -f''(q(t)),$$ which is the definition of high infection risk at $t$.

(2) $\Leftrightarrow$ (3): Because $C$ is constant, the infection rate $\iota(t)$ is decreasing in $\beta$ if and only if $C\iota(t)$ is decreasing in $\beta$. In our expression for $C\iota(t)$, the terms $C$ and $\beta$ only appear within the product $C\beta$. So $C\iota(t)$ is decreasing in $\beta$ if and only if it is decreasing in $C$.
\end{proof}

\begin{proof}[Proof of Proposition~\ref{prop:ctilde}]
(i) We first show $\widetilde{C}(t) < C$ for all $t$. Recall $\widetilde{C}(t) = C - \pi_S(t) + \pi_I(t)$, where $\pi_S(t)$ and $\pi_I(t)$ are the equilibrium continuation payoffs for infected and susceptible individuals. We must show that $\pi_I(t) > \pi_S(t)$.

Susceptible and infected individuals choose action $q(t)<\overline{q}$ while recovered individuals choose action $\overline{q}$. The time at which an individual who is susceptible at time $t$ recovers (which may be $\infty$) first-order stochastically dominates the time at which an individual who is infected at time $t$ recovers. Therefore, the expected flow payoffs from activity are higher for an infected individual than a susceptible individual. Since only susceptible individuals may pay the infection cost $C$, we conclude $\pi_I(t) > \pi_S(t).$

We next show that $(1-\frac{\overline{R}(\infty) - R(t_0)}{1-R(t_0)})C \leq \widetilde{C}(t)$. Recall $\overline{R}(\infty)$ is the share of individuals who are eventually infected when $q(t)=\overline{q}$ for all $t$.

Consider an individual $i$ who is susceptible at time $t_0$. We claim that if $i$ chose $q_i(t)=\overline{q}$ for all $t\geq t_0$, then the probability that $i$ is infected at any time $t \geq t_0$ is at most $\overline{R}(\infty) - R(t)$.

By Lemma~\ref{lem:bound}, the number of individuals are infected at any time $t\geq t_0$ is at most $\overline{R}(\infty)-R(t)$ given any actions $q(t) \leq \overline{q}$. Moreover, the number of individuals are infected at any time $t\geq t_0$ is maximized (conditional on actions before time $t_0$) when $q(t) = \overline{q}$ for all $t \geq t_0$. Since increasing the activity level $q(t)$ chosen by others increases the probability that any given infected individual will infect $i$, this implies the claim.

Now, consider an individual $i$ who is either susceptible (with probability $\frac{S(t)}{S(t)+I(t)}$) or infected (with probability  $\frac{I(t)}{S(t)+I(t)}$) at time $t_0$. Individual $i$ could deviate to choose the same distribution of actions as an individual who is infected at time $t_0$. To do so, individual $i$ initially follows the equilibrium but then transitions to choosing activity level  $q(t)=\overline{q}$ at Poisson rate $\kappa$.

This deviation would increase the continuation payoffs for $i$ conditional on being infected at time $t_0$, and therefore must decrease the continuation payoffs for $i$ conditional on being susceptible at time $t_0$. If individual $i$ is not already infected, then $i$ would pay the infection cost $C$ at some future date with probability at most $\frac{\overline{R}(\infty) - R(t_0)}{1-R(t_0)}$. Thus, the equilibrium continuation payoffs $\pi_S(t_0)$ must be at least $\pi_I(t_0)-\frac{\overline{R}(\infty) - R(t_0)}{1-R(t_0)}\cdot C$. Therefore $$\widetilde{C}(t_0)=C-(\pi_I(t_0)-\pi_S(t_0)) \geq \left(1-\frac{\overline{R}(\infty) - R(t_0)}{1-R(t_0)}\right)C.$$

(ii) We want to bound $\pi_I(t) - \pi_S(t)$ for $t$ large. This gap consists of the difference in flow payoffs from economic activity and the potential infection cost.

We have $I(t) \rightarrow 0$ as in the myopic case. By part (i) of the proposition, $\widetilde{C}(t) \leq C$ for all $C$. So by equation~(\ref{eq:actfl}), we have $q(t) \rightarrow \infty$ as $t \rightarrow \infty$. Therefore the gap between flow payoffs from economic activity vanishes as $t\rightarrow \infty$.

The probability of infection at some time after $t$ is $R(\infty) - R(t)$, and this probability converges to zero as $t\rightarrow \infty$ since $R(t)$ is increasing. So the potential infection cost for susceptible individuals vanishes as $t\rightarrow \infty$. We conclude that $\pi_I(t) - \pi_S(t)\rightarrow 0.$

(iii) Let $\epsilon>0$. We claim that $\pi_I(0) - \pi_S(0) < \epsilon$ for $I(0)$ sufficiently small. Given a discount rate $r$, we can choose $t_0$ large enough so that for any $I(0)$ the contribution to $\pi_I(0)$ from payoffs at time $t\geq t_0$ is at most $\epsilon/3$.

We can then choose $I(0)$ small enough so that $C(I(t_0)+R(t_0)) \leq \epsilon/3$. To see this, note that $I(t_0)+R(t_0)$ is bounded above by its value when $q(t)=\overline{q}$ for all $t$.

Finally, by equation~(\ref{eq:actfl}), we observe that $\max_{0\leq t \leq t_0} q(t) \rightarrow \overline{q}$ as $I(0) \rightarrow 0$. Therefore, we conclude that the gap between flow payoffs from economic activity at times $t \leq t_0$ vanishes as $I(0) \rightarrow 0$. So for $I(0)$ sufficiently small, the total gap in flow payoffs at times $t \leq t_0$ is at most $\epsilon/3$. Combining these three inequalities, we conclude that $\pi_I(0) - \pi_S(0) < \epsilon$.
\end{proof}

\begin{proof}[Proof of Theorem~\ref{thm:prevalenceIR}$'''$]
The proof is the same as the proof of Theorem~\ref{thm:prevalenceIR}, with $C$ replaced by $\widetilde{C}(t)$.
\end{proof}

\begin{proof}[Proof of Theorem~\ref{thm:policy}$'''$]

(1) $\Leftrightarrow$ (2): The proof is the same as the proof of Theorem~\ref{thm:policy}, with $C$ replaced by $\widetilde{C}(t)$.

(2) $\Leftrightarrow$ (3): The infection rate $\iota(t)$ is decreasing in the transmission rate $\beta$ at time $t$ if and only if $\widetilde{C}(t)\iota(t)$ is decreasing in $\beta$. Expanding the product $\widetilde{C}(t)\iota(t)$, the terms $\widetilde{C}(t)$ and $\beta$ only appear within the product $\widetilde{C}(t)\beta$. So $\widetilde{C}(t)\iota(t)$ is decreasing in $\beta$ if and only if it is decreasing in $\widetilde{C}(t)$.

An instantaneous shock to $C$ at time $t$ affects $\widetilde{C}(t)=C+\pi_I(t)-\pi_S(t)$ only via the first term and does not change the difference in continuation payoffs $\pi_I(t)-\pi_S(t)$. So $\widetilde{C}(t)\iota(t)$ is decreasing in $\widetilde{C}(t)$ if and only if it is decreasing in $C$.
\end{proof}

\end{spacing}
\end{document}